\documentclass[poms,nonblindrev]{poms1} 

\OneAndAHalfSpacedXI 

\usepackage{amssymb}
\usepackage{amsthm}
\usepackage{bbm}

\usepackage{xfrac}
\usepackage{xspace}
\newcommand\ie{i.\,e.\xspace}
\newcommand\eg{e.\,g.\xspace}

\newcommand\cf{cf.\xspace}

\newcommand\A{$\mathcal{A}$\xspace}
\newcommand\B{$\mathcal{B}$\xspace}

\usepackage{cleveref}
\usepackage{csquotes}

\usepackage{enumitem}

\usepackage{siunitx}
\def\sym#1{\ifmmode^{#1}\else\(^{#1}\)\fi}

\usepackage{etoolbox}
\robustify\bfseries

\newcommand{\norm}[1]{\left\lVert#1\right\rVert}

\usepackage{booktabs}
\usepackage{subcaption}
\usepackage{array}
\usepackage[section]{placeins}
\usepackage{float}
\usepackage{pbox}
\usepackage{mathtools}

\makeatletter
\renewcommand{\fps@figure}{H}         
\renewcommand{\fps@table}{H}         
\makeatother 

\newcolumntype{x}[1]{>{\centering\arraybackslash\hspace{0pt}}p{#1}}
\newcolumntype{C}[1]{>{\centering\arraybackslash}p{#1}}

\usepackage{tikz}
    \usetikzlibrary{shapes,arrows}
    \usetikzlibrary{positioning}
    \usetikzlibrary{decorations.pathreplacing}
\usepackage{pgfplots}
  \pgfplotsset{compat=1.9}
    \usepgfplotslibrary{statistics}
\usepackage{pgfplotstable}

\tikzstyle{block} = [rectangle, draw, text width=7em, text centered, rounded corners, minimum height=3.5em]
\tikzstyle{wideblock} = [rectangle, draw, text width=9em, text centered, rounded corners, minimum height=3.5em]
\tikzstyle{medblock} = [rectangle, draw, text width=11em, text centered, rounded corners, minimum height=2.0em]
\tikzstyle{tinyblock} = [rectangle, draw, text width=7em, text centered]
\tikzstyle{line} = [draw, -latex']
\tikzstyle{tinyblockdashed} = [rectangle, dashed, draw, text width=7em, text centered]
\tikzstyle{line} = [draw, -latex']

\makeatletter
\DeclareRobustCommand{\rvdots}{%
  \vbox{
    \baselineskip4\p@\lineskiplimit\z@
    \kern-\p@
    \hbox{.}\hbox{.}\hbox{.}
  }}
\makeatother

\usepackage{natbib}
 \bibpunct[, ]{(}{)}{,}{a}{}{,}%
 \def\bibfont{\small}%
\EquationsNumberedThrough    

\usepackage{algorithm}
\usepackage[noend]{algpseudocode}
\algrenewcommand\algorithmicrequire{\textbf{Input:}}
\algrenewcommand\algorithmicensure{\textbf{Output:}}
\usepackage{amsmath,amssymb}
\renewcommand{\norm}[1]{\left\lVert#1\right\rVert}

\makeatletter
\def\BState{\State\hskip-\ALG@thistlm}
\makeatother

\usepackage{xpatch}
\usepackage{multirow}

\makeatletter
\xpatchcmd{\algorithmic}{\itemsep\z@}{\itemsep=1ex}{}{}
\makeatother

\algdef{SE}[SUBALG]{Indent}{EndIndent}{}{\algorithmicend\ }%
\algtext*{Indent}
\algtext*{EndIndent}

\newtheorem{remark}{Remark}[section]
\theoremstyle{definition}
\newtheorem{definition}{Definition}[section]
\theoremstyle{theorem}
\newtheorem{theorem}{Theorem}[section]
\theoremstyle{lemma}
\newtheorem{lemma}{Lemma}[section]

\begin{document}


\RUNAUTHOR{Senoner et al.}

\RUNTITLE{{Addressing distributional shifts in operations management: The case of order fulfillment in customized production}}

\TITLE{Addressing distributional shifts in operations management:\\ The case of order fulfillment in customized production}

\ARTICLEAUTHORS{%
\AUTHOR{Julian Senoner}
\AFF{ETH Zurich, \EMAIL{jsenoner@ethz.ch}} 
\AUTHOR{Bernhard Kratzwald}
\AFF{ETH Zurich, \EMAIL{bkratzwald@ethz.ch}} 
\AUTHOR{Milan Kuzmanovic}
\AFF{ETH Zurich, \EMAIL{mkuzmanovic@ethz.ch}} 
\AUTHOR{Torbjørn H. Netland}
\AFF{ETH Zurich, \EMAIL{tnetland@ethz.ch}} 
\AUTHOR{Stefan Feuerriegel}
\AFF{LMU Munich, \EMAIL{feuerriegel@lmu.de}}
} 

\ABSTRACT{
\noindent
To meet order fulfillment targets, manufacturers seek to optimize production schedules. Machine learning can support this objective by predicting throughput times on production lines given order specifications. However, this is challenging when manufacturers produce customized products because customization often leads to changes in the probability distribution of operational data---so-called \emph{distributional shifts}. Distributional shifts can harm the performance of predictive models when deployed to future customer orders with new specifications. The literature provides limited advice on how such distributional shifts can be addressed in operations management. Here, we propose a data-driven approach based on adversarial learning and job shop scheduling, which allows us to account for distributional shifts in manufacturing settings with high degrees of product customization. We empirically validate our proposed approach using real-world data from a job shop production that supplies large metal components to an oil platform construction yard. Across an extensive series of numerical experiments, we find that our adversarial learning approach outperforms common baselines. Overall, this paper shows how production managers can improve their decision-making under distributional shifts. 
}

\KEYWORDS{machine learning, adversarial learning, distributional shifts, order fulfillment, manufacturing} 

\maketitle

%

\sloppy
\raggedbottom

\section{Introduction} 
\label{sec:Introduction}

Order fulfillment is a crucial performance indicator in manufacturing \citep{Song.1999}. Achieving on-time delivery can be particularly difficult when manufacturers produce highly customized products \citep{Cohen.2003}. For these manufacturers, incoming customer orders can involve multiple non-standard production tasks that must be planned effectively. This is challenging because order throughput times vary across customer specifications. Completing production orders too early causes unnecessary inventory costs and increases the risk of rework in case of engineering changes, whereas delays diminish service levels and can lead to substantial economic losses \citep{Cohen.2003}. Therefore, reducing deviations from order delivery due dates is important for maintaining a good cost performance. Motivated by this objective, manufacturers continuously seek to improve their planning accuracy.

Machine learning can support manufacturers in achieving high planning accuracy. Specifically, machine learning can improve managerial decision-making by utilizing historical observations to predict throughput times given order specifications \citep[\cf][]{Grabenstetter.2014}. This enables production managers to optimize their production plans and reduce costs by mitigating deviations from order delivery due dates. However, implementing machine learning in manufacturing settings with high degrees of product customization is challenging because orders may involve unique specifications and high variety (\eg, producing customized components for one-of-a-kind oil platforms). Due to heterogeneity between different customer orders, such settings typically generate operational data for which probability distributions of different customer orders are highly dissimilar. In the machine learning literature, this is commonly referred to as \emph{distributional shifts}. Consequently, the standard assumption of identically distributed samples in predictive analytics \citep{Hastie.2009} is violated, and the prediction performance may deteriorate.

Scholars have argued that distributional shifts are a key challenge of successfully applying predictive analytics in management \citep{Simester.2020}. While there is increasing traction of predictive analytics in Operations Management~(OM) \citep[\eg,][]{Baardman.2018, Misic.2020, Olsen.2020, Bastani.2021}, there is, however, scarce advice on how distributional shifts can be addressed. Existing approaches to address distributional shifts make use of model retraining \citep{Cui.2018} or transfer learning via fine-tuning \citep{Pan.2010, McNamara.2017, Kouw.2018, Bastani.2020a}. These methods are effective in make-to-stock manufacturing (\eg, fast-moving consumer goods) characterized by low variety and high volumes. Yet, they are likely to fall short in manufacturing settings characterized by high variety and low volumes. Both model retraining and transfer learning via fine-tuning require labeled data for both historical and forthcoming orders, yet labels for the latter are not available when dealing with new orders. As such, tailored approaches for customized production are needed. 

To meet order fulfillment targets in customized production, we develop a data-driven approach to predict order throughput times and then perform job shop scheduling. Key to our approach is that we predict when the order will be finished and, then, use this information to optimize scheduling decisions. Due to the operational heterogeneity in customized production, making predictions of throughput times is challenging as there are distributional shifts between different customer orders, which violate the standard assumption of machine learning and which can reduce prediction performance. In our approach, we account for distributional shifts between different customer orders through the use of adversarial learning. Adversarial learning \citep{Goodfellow.2014} is a recent innovation in artificial intelligence to make inferences under two opposing---thus adversarial---objectives. In our case, we leverage adversarial learning to combine the following two objectives: (1)~to predict throughput times with the best possible performance and, simultaneously, (2)~to minimize the distance of the neural network representations of the operational data between the historical and the forthcoming order setting. Importantly, the latter accounts for the new specifications of the forthcoming order setting (without knowing the ground-truth labels) and, as a result, yields better predictions of throughput times and better scheduling decisions for future orders. 

We evaluate the effectiveness of our proposed approach in a series of numerical experiments using real-world industrial data from Aker Solutions (from here on Aker), a leading engineering company in the energy sector. We focus on a job shop production that supplies large metal components for the construction of oil platforms. In this case, components produced for different customer order settings (\ie, new oil platforms) involve many idiosyncratic specifications. We find a substantial distributional shift between order settings; that is, the conditions under which the components are produced are highly dissimilar. Especially when starting to produce components with new specifications, it is challenging for na{\"ive} predictions using machine learning to provide accurate estimates of throughput times. This is addressed in our approach based on adversarial learning as it explicitly accounts for the distributional shifts between different order settings. We then compare our proposed approach for job shop scheduling against several data-driven baselines. Across an extensive series of numerical experiments, we find that our approach outperforms the baselines and can offer considerable cost savings.   

This paper makes three main contributions. First, we provide empirical evidence of how predictive analytics can improve order fulfillment in customized production. Second, we contribute to the practice and literature on predictive analytics in OM \citep[\eg,][]{Misic.2020, Olsen.2020, Bastani.2021}. Here, we address the problem of distributional shifts, which is particularly pertinent in settings subject to extensive heterogeneity. Third, we  tailor adversarial learning to an OM context and study its operational value for job shop scheduling. Altogether, our work has direct managerial implications: manufacturers need to identify and cater for distributional shifts in customized production. Simply deploying predictive analytics without addressing distributional shifts may result in subpar decisions. 

This paper is structured as follows. \Cref{sec:background} motivates the challenges of predictive analytics in customized production under distributional shifts, yet revealing a scarcity of methods for addressing distributional shifts in OM. \Cref{sec:aker} introduces our empirical setting at Aker. \Cref{sec:methods} proposes our model for job shop scheduling, which uses adversarial learning to predict throughput times while addressing distributional shifts. \Cref{sec:simulation_study} evaluates our approach in a series of numerical experiments under different operational contexts. We report various robustness checks in \Cref{sec:robustness_checks}. Based on our results, \Cref{sec:implications} discusses the implications for making robust inferences under distributional shifts in OM. 


\section{Related work} 
\label{sec:background}

This paper is motivated by the practical problem of making robust inferences under distributional shifts and, for this purpose, draws upon statistical learning theory. We see four streams of research as particularly relevant to this work: (1)~predictive analytics in OM, (2)~job shop scheduling, (3)~distributional shifts, and (4)~adversarial learning.

\subsection{Predictive analytics in OM}
\label{sec:predictive_analytics_om}

Predictive analytics can support managers in making decisions by modeling uncertain operational outcomes \citep{Choi.2018,Cohen.2018,Feuerriegel.2022}. Recent methodological advances, accompanied by the increasing availability of data, have accelerated the adoption of predictive analytics in OM \citep{Misic.2020, Olsen.2020, Bastani.2021, jakubik2022}. In the following, we provide an overview of predictive analytics in OM. For a detailed review of the literature, see \cite{Misic.2020} and \citet{Bastani.2021}.

There are many promising demonstrations of predictive analytics in OM, such as sales and demand forecasting \citep{Carbonneau.2008, Ferreira.2016, Baardman.2018, Cui.2018, Lau.2018}, revenue management  \citep{Bernstein.2018, Feldman.2021, Chen.2022}, location selection problems \citep{Glaeser.2019, Huang.2019}, last mile delivery \citep{Liu.2020}, product recall decisions \citep{Mukherjee.2018}, inventory management \citep{Bertsimas.2019}, procurement under demand uncertainty \citep{Ban.2019}, and the estimation of the remaining useful life of products \citep{Mazhar.2007}. 

In manufacturing operations, \citet{Senoner.2022} adapted predictive analytics to improve process quality. However, despite their relevance, applications of predictive analytics and decision-making are still scarce in manufacturing. OM scholars have therefore argued that more attention in this \textquote{under-researched} area is needed \citep{Feng.2018}. This particularly holds true for manufacturing settings, where products are manufactured with high variety and in low volumes. These characteristics make it particularly challenging to apply conventional prediction methods due to distributional shifts. Here, our paper contributes to the OM literature by tailoring predictive analytics for manufacturing settings with high degrees of product customization. 

Applying predictive analytics to data-scarce settings (\eg, new products) is a known challenge in OM. \citet{Kesavan.2020} compare analytics and expert decision-making in a field experiment, finding that giving discretionary power to experts is beneficial in growth-stage products. Instead of experts, practitioners can also revert to social media \citep{Cui.2018} or other proxies \citep{Bastani.2020a}. However, our work is different in that there is no such data available for forthcoming orders in customized production. As a remedy, we develop a data-driven approach based on adversarial learning and later demonstrate its operational value.

\subsection{Job shop scheduling} 
\label{sec:job_shop_scheduling}

Job shop scheduling problems consider a multitude of arriving production orders that compete for processing time on common resources \citep{Adams.1988, Wein.1992}. These orders are typically associated with an arrival date $t$, a due date $d$, and a stochastic throughput time $y$. The difference between $d$ and $t$ defines the planned leadtime $y^{(plan)}$. The problem is to schedule the individual production orders by optimizing against a time-dependent objective. A common objective is to minimize the total cost of (1)~\emph{earliness}, \ie, completing production orders  too early (before $d$) and (2)~\emph{tardiness}, \ie, completing production orders too late (after $d$) \citep[\eg,][]{Seidmann.1981, Bagchi.1994, Federgruen.1996, Atan.2016}. A simple form of the total cost function $C$ is given by
\begin{equation}
\SingleSpacedXI
\label{eqn:job_shop}
C(y, y^{(plan)}) = \zeta^{(early)}[y^{(plan)} - y]^{+} + \zeta^{(tardy)}[y - y^{(plan)}]^{+} ,
\end{equation}
where $[ \cdot ]^+$ returns the positive part of an expression and where the convex functions $\zeta^{(early)}(\cdot)$ and $\zeta^{(tardy)}(\cdot)$ measure the costs of earliness and tardiness per time unit \citep{Seidmann.1981}. As can be seen in Eq.~(\ref{eqn:job_shop}), the planning accuracy in a job shop production strongly depends on accurate time estimates for throughput times. 

The throughput time $y$ is typically estimated based on historical observations. Recent approaches have also included covariate information for the estimation of throughput times \citep{Grabenstetter.2014}. The common assumption is that data from previous customer orders are sampled from the same probability distribution as the data from forthcoming customer orders. This rarely holds in manufacturing settings with high degrees of product customization. Customization often leads to changes in the distribution of operational data, which makes it particularly challenging to provide accurate throughput time estimates. This paper seeks to predict throughput times while accounting for distributional shifts in the operational data due to product customization. 

\subsection{Distributional shifts} 
\label{sec:domain_shift}

Predictive analytics deals with the problem of inference; that is, analyzing patterns and making predictions from observational data \citep{Hastie.2009, Ghahramani.2015}. Formally, one infers an outcome $y \in Y$ based on an input $x \in X$ from a predictive model $f$  \citep{Jordan.2015}. As formalized in statistical learning theory, the performance of a predictive model depends on its ability to generalize well on out-of-sample observations. The core assumption is that both past and future observations are independently drawn from the same probability distribution over $X \times Y$. If this assumption is violated (\ie, $\mathbb{P}(X,Y) \neq \mathbb{P}(X',Y')$ between in-sample and out-of-sample observations), the performance of predictive models is likely to deteriorate for future observations. This is the case if a predictive model is deployed on data that stem from a different probability distribution than the training data. In the literature, this is referred to as \emph{distributional shift} \citep[or domain shift;][]{Kouw.2018}.

Distributional shifts have been extensively studied for the specific requirements in computer vision and computational linguistics. These studies mostly draw upon common datasets for benchmarking (\eg, handwritten digits) but do not involve operational data. Recently, there have been some technical contributions focusing on learning under distributional shifts \citep[\eg,][]{Ganin.2016, Tzeng.2017, Shen.2018, Wang.2019}. These methods can be subsumed under the term \textquote{domain adaptation} or \textquote{domain adaptive learning.} The objective of domain adaptive learning is to perform an end-to-end prediction task while simultaneously considering distributional shifts between two different domains (\eg, predicting product reviews for books based on product reviews written for movies). While domain adaptive learning has shown great potential with image and text data, its operational value in OM practice has not yet been studied. 

There are approaches for handling distributional shifts but with a clearly different objective \citep{Pan.2010, McNamara.2017,Kouw.2018}. First, there is \emph{model retraining} \citep{Cui.2018}. Here, a model is (continuously) updated using data $(x,y)$ with features and labels from the new operational setting, so that it adapts to the distributional shift in operational data. As in online learning, the latter requires continuous access to labels, which makes it suitable for make-to-stock manufacturing (\eg, fast-moving consumer goods). In contrast, such labels are not available for forthcoming orders in customized production because of which this approach is not applicable to our work. Second, there is \emph{transfer learning} via fine-tuning \citep{McNamara.2017, Kouw.2018}. Here, inferences are made between different predictive tasks (\eg, changing $y$ from fault risk to cost) or between different populations. To this end, data from the new predictive task or new population is used to update the model weights. This is effective for operational settings with proxy data \citep{Bastani.2020a}. Yet, such data are typically unavailable for forthcoming orders in customized production. In sum, the above methods require labeled data from the deployment setting; therefore, none of these methods fulfill the objective of this work.

Distributional shifts represent a key hurdle for applying predictive analytics. Yet, despite its relevance, there is a scarcity of research on its practical implications \citep{Simester.2020}. To the best of our knowledge, there is no research that suggests how to effectively mitigate distributional shifts in manufacturing. However, distributional shifts appear in all real-world operations and can lead to poor decision-making. Motivated by the general trend in manufacturing toward increased product customization \citep{Feng.2018, Olsen.2020, Choi.2021}, this paper addresses distributional shifts in operational data through the use of adversarial learning.

\subsection{Adversarial learning} 
\label{sec:adversarial_learning}

The term \textquote{adversarial learning} refers to a general technique in predictive analytics whereby a neural network is supposed to learn two adversarial objectives \citep{Goodfellow.2014}. For example, it allows one to train a neural network that has good prediction performance \emph{and} where the representation of the neural network simultaneously fulfills another constraint. Mathematically, the two adversarial objectives can be viewed as a two-player minimax game. Yet, implementing adversarial learning in practice is challenging. On the one hand, an appropriate optimization technique must be chosen to ensure convergence with state-of-the-art solvers (\eg, stochastic gradient descent), and, on the other hand, optimization of the two objectives must take place in the latent space of the neural network parameters. 

Adversarial learning has first been introduced as part of generative adversarial networks \citep[GANs;][]{Goodfellow.2014}. GANs generate synthetic data (\eg, artificially created images) that are indistinguishable from samples drawn from a real data distribution (\eg, distribution of real images). In order to achieve this, two separate two neural networks are used: a generator $G$ that generates a new synthetic sample and a discriminator $D$ that estimates if a sample stems from the original distribution (\ie, sampled as in the training data) or from the model distribution (\ie, sampled from the generator). Both networks are then trained jointly with adversarial objectives: the discriminator is trained to distinguish the two distributions, while the generator model is trained to fool the discriminator. This results in a two-player minimax game with an optimal solution in which the generator distribution is equal to the data distribution. 

Adversarial learning has also shown success in domain adaptive learning for image and text data \citep{Ganin.2016, Tzeng.2017, Shen.2018}. Here,  adversarial learning is applied to map features from different domains into a common latent space, so that the training procedure becomes an end-to-end prediction task. This is usually realized by training one neural network to achieve the best possible prediction performance, while a second adversarial network is trained to keep the feature distributions from both domains close. This idea is based on the theoretical findings of \cite{bendavid.2007, bendavid.2010}, which suggest that, for good feature representations in cross-domain transfer, a discriminator should not be able to distinguish from which domain an observation originated. 

In this paper, we tailor domain adversarial learning to OM decision-making; that is, we address distributional shifts between different orders in customized production to improve job shop scheduling. 

\section{Empirical setup at Aker} 
\label{sec:aker}

\subsection{Job shop production at Aker} 
\label{sec:ETO_manufacturing_at_Aker}

Our empirical application is carried out at Aker, headquartered in Lysaker, Norway. Aker is a leading engineering company in the energy sector, with an annual turnover of approximately USD 3.4 billion in 2021. The company covers the entire value chain, including fabrication engineering, purchasing, manufacturing, and delivery. We focus on a job shop production involving customer orders for large metal components that are supplied for the construction of oil platforms. Due to strong dependencies, delays in individual production orders can lead to substantial economic losses. Therefore, Aker puts great emphasis on order fulfillment in its component production sites. 

In our numerical experiments, we use data from the two most recent order settings: \textquote{Johan Castberg Floating Production Vessel} (setting~\A) and \textquote{Johan Sverdrup Riser Platform Modification} (setting~\B). Setting~\A involves the production of topside modules for a floating production vessel. Setting~\B involves the production of components for the modification of an offshore oil platform. As can be observed in \Cref{fig:order_images}, the two settings for which Aker supplies components differ radically. Both settings contain complex piping networks that consist of thousands of interconnected metal spools. Each of these spools corresponds to an individual production order that can have distinctive specifications in terms of material requirements, size, and shape. To avoid costly delays at the construction yard of the customer, it is crucial that the spools are available on time. 

\begin{figure}
\centering
\captionsetup[sub]{font=scriptsize}
\begin{subfigure}{.35\textwidth}

  \centering
  \caption{\raggedright\scriptsize{Setting~\A: Johan Castberg Floating Production Vessel (Credit: Equinor)}}
  \includegraphics[height=3cm]{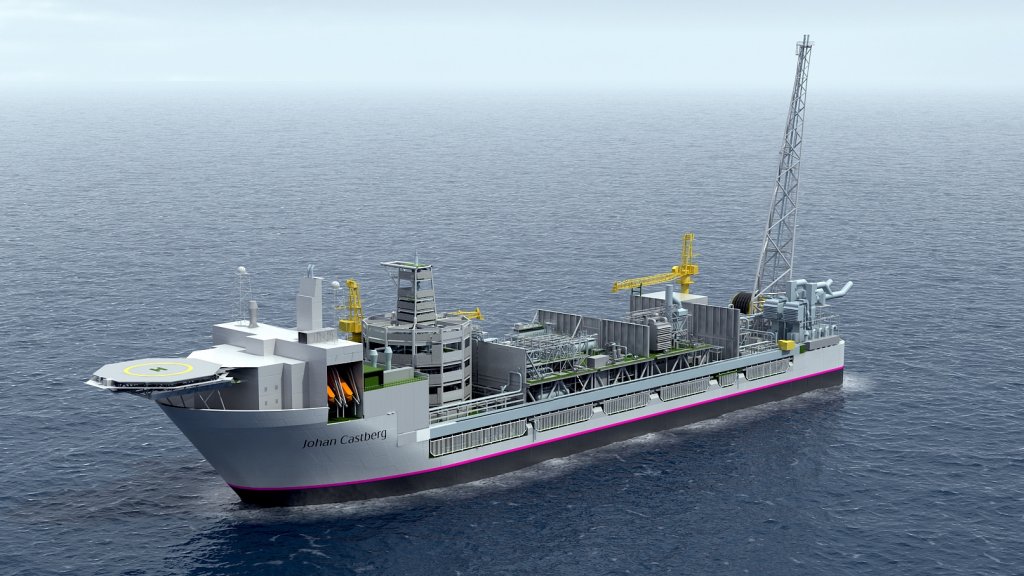}
\end{subfigure}%
\begin{subfigure}{.35\textwidth}
  \centering
  \caption{\raggedright\scriptsize{Setting~\B: Johan Sverdrup Riser Platform Modification (Credit: Anette Westgård/Equinor)}}
  \includegraphics[height=3cm]{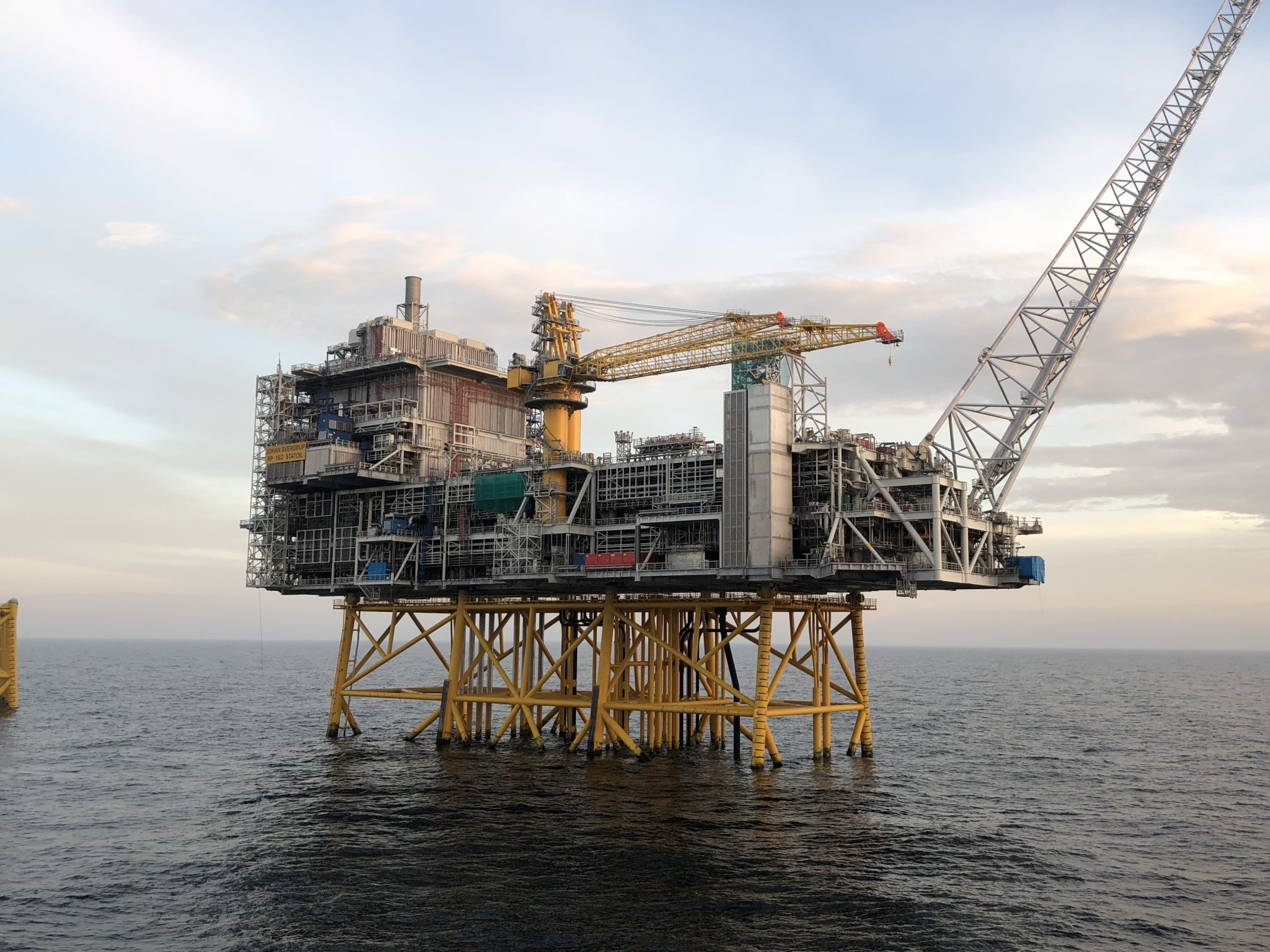}
\end{subfigure}
\begin{subfigure}{.25\textwidth}
  \centering
  \caption{\scriptsize{Spool and pipe production (Credit: Aker)}}
  \includegraphics[height=3cm]{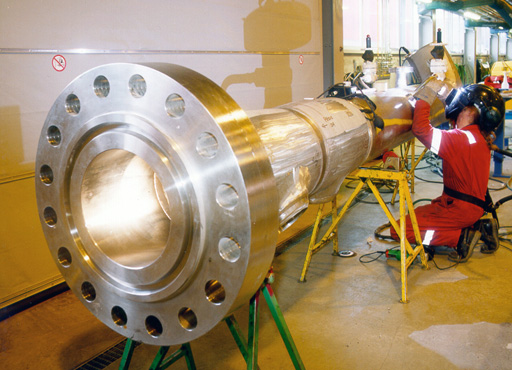}
\end{subfigure}

\vspace{0.01cm}
\caption{Empirical context at Aker.} 
\label{fig:order_images}
\end{figure}

The order due dates of all spools are recorded in a centralized enterprise resource planning (ERP) system. The production management at Aker reviews the spools for which the engineering specifications and raw materials have arrived; that is, the backlog. This provides the basis for deciding which spools should be started next. Tardiness in the spool production (\ie, the actual finish date exceeds the order due date) can cause substantial disruptions in the value chain. In contrast, spools that finish too early result in excess inventory that must be managed in a limited outdoor space (recall, the spools are large, bulky, and heavy) and may lead to rework in the event of engineering change orders. For Aker, the cost of carrying excess inventory is substantial. By accurately estimating the throughput times of orders in the backlog, production managers can optimize their scheduling decisions to improve the on-time delivery performance (see \Cref{fig:timelines}).

\begin{figure}[htb]
\centering
\begin{subfigure}{.5\textwidth}
  \centering
  \caption{\footnotesize{Early spool}}
  \tikzset{every picture/.style={line width=0.75pt}} 

\begin{tikzpicture}[x=0.60pt,y=0.6pt,yscale=-1,xscale=1]
\footnotesize

\draw   (200,180) -- (500,180) -- (500,210) -- (200,210) -- cycle ;
\draw [line width=0.75]    (200,152.5) -- (200,177) ;
\draw [shift={(200,180)}, rotate = 270] [fill={rgb, 255:red, 0; green, 0; blue, 0 }  ][line width=0.08]  [draw opacity=0] (8.93,-4.29) -- (0,0) -- (8.93,4.29) -- cycle    ;
\draw  [color={rgb, 255:red, 0; green, 0; blue, 0 }  ,draw opacity=1 ][fill={rgb, 255:red, 155; green, 155; blue, 155 }  ,fill opacity=0.45 ] (270,180) -- (420,180) -- (420,210) -- (270,210) -- cycle ;
\draw [line width=0.75]    (270,152.5) -- (270,177) ;
\draw [shift={(270,180)}, rotate = 270] [fill={rgb, 255:red, 0; green, 0; blue, 0 }  ][line width=0.08]  [draw opacity=0] (8.93,-4.29) -- (0,0) -- (8.93,4.29) -- cycle    ;
\draw [line width=0.75]    (420,152.5) -- (420,177) ;
\draw [shift={(420,180)}, rotate = 270] [fill={rgb, 255:red, 0; green, 0; blue, 0 }  ][line width=0.08]  [draw opacity=0] (8.93,-4.29) -- (0,0) -- (8.93,4.29) -- cycle    ;
\draw [line width=0.75]    (500,153) -- (500,177.5) ;
\draw [shift={(500,180.5)}, rotate = 270] [fill={rgb, 255:red, 0; green, 0; blue, 0 }  ][line width=0.08]  [draw opacity=0] (8.93,-4.29) -- (0,0) -- (8.93,4.29) -- cycle    ;

\draw (200,112) node [text width=3cm, anchor=north][align=center]{\baselineskip=10pt Arrival in \\ backlog \par};

\draw (270,112) node [text width=3cm, anchor=north][align=center]{\baselineskip=10pt Actual \\ start date \par};

\draw (420,112) node [text width=3cm, anchor=north][align=center]{\baselineskip=10pt Actual \\ finish date \par};

\draw (500,112) node [text width=3cm, anchor=north][align=center]{\baselineskip=10pt Order \\ due date \par};

\draw [|-|] (200,225) -- (270,225);
\draw (235,230) node [text width=3cm, anchor=north][align=center]{\baselineskip=10pt Backlog \par};

\draw [|<-|, >=latex] (420,225) -- (500,225);
\draw (460,230) node [text width=3cm, anchor=north][align=center]{\baselineskip=10pt Earliness \par};

\draw [|-|] (270,100) -- (420,100);
\draw (350,70) node [text width=3cm, anchor=north][align=center]{\baselineskip=10pt Throughput time \par};

\end{tikzpicture}
\end{subfigure}%
\begin{subfigure}{.5\textwidth}
  \centering
  \caption{\footnotesize{Tardy spool}}
  \tikzset{every picture/.style={line width=0.75pt}} 

\begin{tikzpicture}[x=0.60pt,y=0.6pt,yscale=-1,xscale=1]
\footnotesize

\draw   (200,180) -- (500,180) -- (500,210) -- (200,210) -- cycle ;
\draw [line width=0.75]    (200,152.5) -- (200,177) ;
\draw [shift={(200,180)}, rotate = 270] [fill={rgb, 255:red, 0; green, 0; blue, 0 }  ][line width=0.08]  [draw opacity=0] (8.93,-4.29) -- (0,0) -- (8.93,4.29) -- cycle    ;
\draw  [color={rgb, 255:red, 0; green, 0; blue, 0 }  ,draw opacity=1 ][fill={rgb, 255:red, 155; green, 155; blue, 155 }  ,fill opacity=0.45 ] (270,180) -- (500,180) -- (500,210) -- (270,210) -- cycle ;
\draw [line width=0.75]    (270,152.5) -- (270,177) ;
\draw [shift={(270,180)}, rotate = 270] [fill={rgb, 255:red, 0; green, 0; blue, 0 }  ][line width=0.08]  [draw opacity=0] (8.93,-4.29) -- (0,0) -- (8.93,4.29) -- cycle    ;
\draw [line width=0.75]    (420,152.5) -- (420,177) ;
\draw [shift={(420,180)}, rotate = 270] [fill={rgb, 255:red, 0; green, 0; blue, 0 }  ][line width=0.08]  [draw opacity=0] (8.93,-4.29) -- (0,0) -- (8.93,4.29) -- cycle    ;
\draw [line width=0.75]    (500,153) -- (500,177.5) ;
\draw [shift={(500,180.5)}, rotate = 270] [fill={rgb, 255:red, 0; green, 0; blue, 0 }  ][line width=0.08]  [draw opacity=0] (8.93,-4.29) -- (0,0) -- (8.93,4.29) -- cycle    ;

\draw [-] (420,180) -- (420,210);

\draw (200,112) node [text width=3cm, anchor=north][align=center]{\baselineskip=10pt Arrival in \\ backlog \par};

\draw (270,112) node [text width=3cm, anchor=north][align=center]{\baselineskip=10pt Actual \\ start date \par};

\draw (420,112) node [text width=3cm, anchor=north][align=center]{\baselineskip=10pt Order \\ due date \par};

\draw (500,112) node [text width=3cm, anchor=north][align=center]{\baselineskip=10pt Actual \\ finish date \par};

\draw [|-|] (200,225) -- (270,225);
\draw (235,230) node [text width=3cm, anchor=north][align=center]{\baselineskip=10pt Backlog \par};

\draw [|->|, >=latex] (420,225) -- (500,225);
\draw (460,230) node [text width=3cm, anchor=north][align=center]{\baselineskip=10pt Tardiness \par};

\draw [|-|] (270,100) -- (500,100);
\draw (385,70) node [text width=3cm, anchor=north][align=center]{\baselineskip=10pt Throughput time \par};

\end{tikzpicture}
\end{subfigure}
\caption{Example timelines for earliness and tardiness in spool production.}
\label{fig:timelines}
\end{figure}
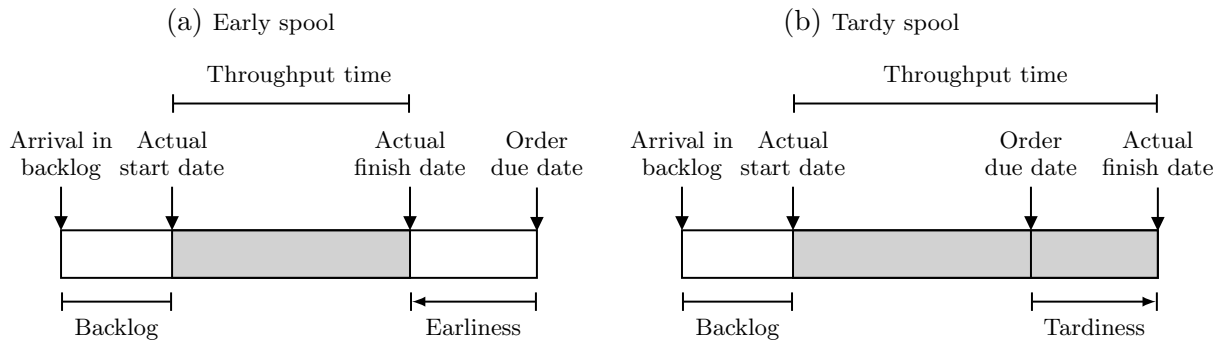

To meet order fulfillment targets, Aker follows a two-staged, predict-then-optimize approach: In the first stage, Aker estimates the throughput times for customer orders, and, in the second stage, production managers solve a job shop scheduling problem to reduce costs from early and tardy production orders (see \Cref{sec:job_shop_scheduling} for an overview on job shop scheduling). However, because each order setting (\ie, construction of oil platform) served by Aker is unique, there is substantial heterogeneity between different customer orders and thus distributional shifts in the operational data. Therefore, a na{\"i}ve application of machine learning to predict throughput times may lead to a suboptimal prediction performance and eventually suboptimal scheduling decisions. This motivates our data-driven approach to address distributional shifts between order settings with adversarial learning.

\subsection{Empirical task} 
\label{sec:prediction_task}

The task is to solve a job shop scheduling problem that optimizes decision-making such that the cost of deviations from the order due dates is minimized. Specifically, we aim to schedule orders for setting~\B (``Johan Sverdrup Riser Platform Modification'') while making use of historical data from setting~\A (``Johan Castberg Floating Production Vessel''). 

Aker provided us with operational data from the two order settings. The data for setting~\A comprise the production details of $n = 5,830$ spools that were produced between January 2019 and September 2019. The data for setting~\B comprise the production details of $m = 3,866$ spools that were produced between September 2019 and April 2020. Notably, there is no chronological overlap between the operational data from setting~\A and setting~\B (see \Cref{fig:order_timelines}). 

For job shop scheduling, we consider all information available to Aker at the time of scheduling the orders for the forthcoming setting ~\B; that is, we use historical data from setting~\A to schedule the orders belonging to setting~\B. Hence, for both settings, we have access to spool-specific features. We denote these features by $\{(x^\mathcal{A}_i)\}_{i=1}^n$ and $\{(x^\mathcal{B}_i)\}_{i=1}^m$, respectively. For the historical setting~\A, we additionally have information on the actual order throughput times that were observed in the past. We refer to them as ``labels'' and denote them by $\{(y^\mathcal{A}_i)\}_{i=1}^n$. In contrast, for setting~\B, we do not have such labels with information on order throughput times, because this is the forthcoming order setting for which we predict throughput times and schedule individual orders. In the following, we make use of the order throughput times for setting~\B but only for the purpose of evaluation.

\begin{figure}
\centering
\begin{tikzpicture}
\footnotesize

\definecolor{color0}{rgb}{1,0.498039215686275,0.0549019607843137}
\definecolor{color1}{rgb}{0.12156862745098,0.466666666666667,0.705882352941177}

\begin{axis}[
width=0.75\textwidth,
height = 3.25 cm,
tick align=outside,
tick pos=left,
x grid style={white!69.0196078431373!black},
xmin=1, xmax=17,
xtick style={color=black},
y grid style={white!69.0196078431373!black},
ymin=-0.25, ymax=2,
xtick ={2,4,6,8,10,12,14,16},
xticklabels={Feb/19, Apr/19, Jun/19, Aug/19, Oct/19, Dec/19, Feb/20, Apr/20},
ytick ={0,1},
yticklabels={Setting $\mathcal{B}$, Setting $\mathcal{A}$},
ytick style={color=black}
]
\addplot [line width=10pt, color1]
table {%
1.8 1
9.5 1
};
\addplot [line width=10pt, color0]
table {%
9.5 0
16.6 0
};
\addplot [|-|,line width=0.75pt, black]
table {%
1.8 1.425
9.5 1.425
};

\addplot [|-|, line width=0.75pt, black]
table {%
9.5 0.425
16.6 0.425
};

\draw (45,195) node [text width=4cm][align=center]{\baselineskip=10pt Model estimation $(x^\mathcal{A},y^\mathcal{A})$ \par};

\draw (120,95) node [text width=4cm][align=center]{\baselineskip=10pt Model deployment $(x^\mathcal{B})$ \par};

\end{axis}

\end{tikzpicture}
\caption{Spool production timelines of order settings \A and \B.}
\label{fig:order_timelines}
\end{figure}
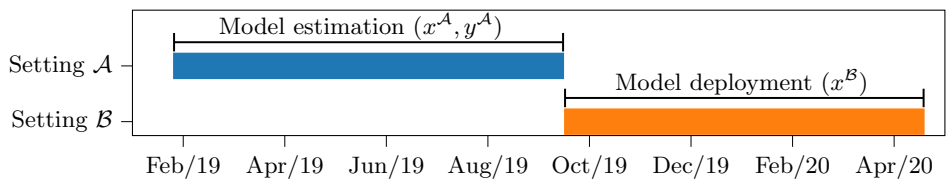

The observed throughput times from setting~\A and setting~\B are shown in \Cref{fig:target_distributions}, suggesting considerable differences. The average throughput time in setting~\A amounts to 32.0 days, whereas the average throughput time in setting~\B amounts to 35.2 days. A Welch's $t$-test confirms that the differences in throughput times are statistically significant ($p<0.001$). Recall again that the throughput times from setting~\B are unknown at the time of prediction (\ie, when starting to produce spools for setting~\B) and are only used for evaluation.

\begin{figure}[htbp]
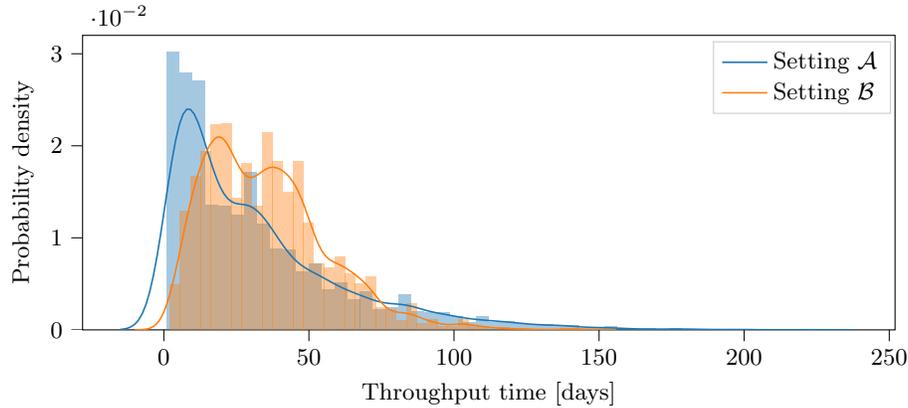

\centering
\include{Figures/target_distribution}
\caption{Distributions of throughput times for setting~\A and setting~\B.}
\label{fig:target_distributions}
\end{figure}

The data from both order settings contain $d=20$ spool-specific features, which we use to predict the throughput times (\Cref{tab:features}). The features include detailed information about the spool-specific timelines, material specifications, and required processing steps. Note that Aker computes the features daily with the available information before and until the time point when the prediction is made. Because of that, a look-ahead bias is prevented. 

All features are potentially subject to distributional shifts between the two order settings. For example, one setting may require spools with particularly fine tolerances that were not needed in previous settings. There may also be distributional shifts when the interrelations between features change. This can, for example, be the case if one order setting requires long and thin spools while the other requires short and thick spools.

\begin{table}[htbp]
\renewcommand*{\arraystretch}{1.15}
\tiny
\centering
\caption{List of features.} 
\label{tab:features}
\begin{tabular}{c p{4.5cm} c p{7.5cm}}
\toprule
\textbf{Feature} & \textbf{Description} & \textbf{Feature} & \textbf{Description}\\
\midrule

$x_{1}$ & Material multiplication factor &  $x_{11}$  & Demolition multiplication factor \\
$x_{2}$ & Insulation thickness around the spool &  $x_{12}$  & Min. design temperature of the fluid/gas in the spool \\
$x_{3}$ & Length of the spool &  $x_{13}$  & Test pressure of the fluid/gas in the spool during testing \\
$x_{4}$ & Dry weight of the spool &  $x_{14}$  & Operational pressure of the gas/fluid in the spool \\
$x_{5}$ & Summed length of all welds of the spool &  $x_{15}$  &  Max. operational temperature of the fluid/gas in the spool \\
$x_{6}$ & Number of welds &  $x_{16}$  & Max. design temperature of the fluid/gas in the spool \\
$x_{7}$ & Average weld diameter &  $x_{17}$  & Number of different materials needed to produce the spool \\
$x_{8}$ & Average weld thickness &  $x_{18}$  & Number of different tasks to produce the spool \\
$x_{9}$ & Planned length of the job &  $x_{19}$  & Max. average historical delay of the material deliveries for the spool \\
$x_{10}$  & Planned sum of workline hours & $x_{20}$ & Revision multiplication factor  \\
\bottomrule
\end{tabular}
\end{table}

\subsection{Exploratory analysis of distributional shifts} 
\label{sec:analysis_of_distributional_shift}

We now explore the distributional shifts between the spool-specific features of setting~\A and setting~\B. Because our operational data comprise $d = 20$ features, we focus on the following multivariate methods for assessing the distributional shifts: (1)~$t$-distributed stochastic neighbor embedding ($t$-SNE) and (2)~adversarial validation. The results are summarized below.

First, we apply $t$-SNE \citep{vanderMaaten.2008} to investigate (dis-)similarities in the feature distributions. The $t$-SNE method is a nonlinear dimensionality reduction technique that is specifically designed for visualizing high-dimensional data. The idea behind the $t$-SNE method is to convert the similarities between data points into joint probabilities and to minimize the Kullback--Leibler divergence between the joint probabilities of a low-dimensional embedding and the original feature space \citep{vanderMaaten.2008}. This yields a low-dimensional representation that can be visualized. 

We utilize $t$-SNE to assess whether the spool-specific features of setting~\A and setting~\B are distributed similarly. \Cref{fig:tsne} shows the feature representations for both settings in a two-dimensional space. Note that the axes do not have a specific meaning but only give an intuition about how the spool observations are distributed in the original feature space (\ie, $d = 20$). It can be observed that the feature representations form largely disjunct clusters with little overlap between the two settings. This provides strong evidence that the operational specifications of setting~\A and setting~\B are substantially different.

\begin{figure}[htbp]
\centering
\includegraphics[width=0.60\linewidth]{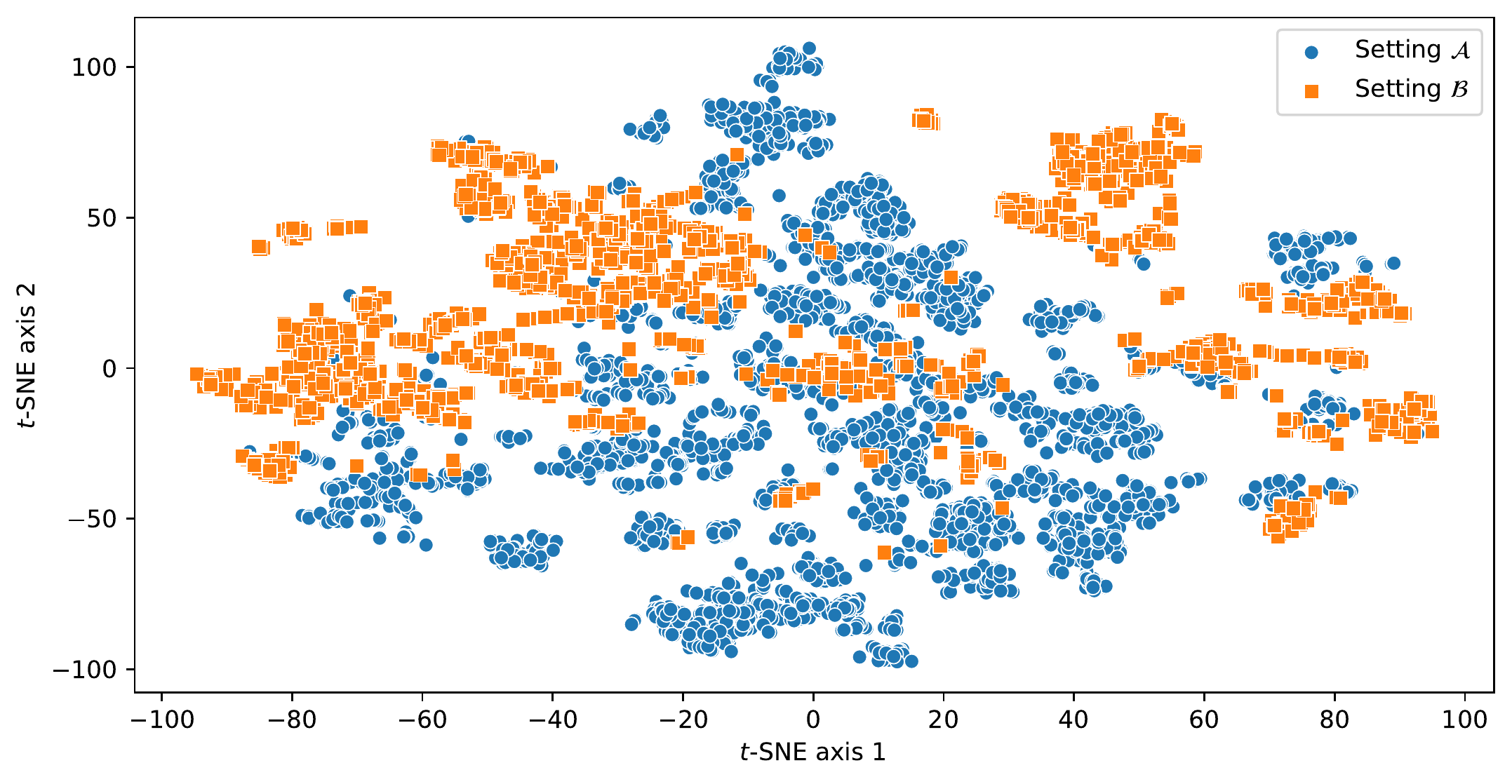}
\caption{Two-dimensional representation of the spool-specific feature spaces based on $t$-SNE.}
\label{fig:tsne}
\end{figure}

Second, we are interested in which features are particularly important in explaining the distributional shifts between both order settings. To achieve this, we draw upon adversarial validation \citep{Pan.2020}. In adversarial validation, one trains a classifier that discriminates between features originating from setting~\A and setting~\B. More formally, we learn a binary classifier to distinguish whether a feature $x$ is drawn from $\mathcal{A}^X$ or $\mathcal{B}^X$. The labels are given by binary indicators that suggest the setting from which $x$ was sampled. Provided there is no distributional shift between the two settings, a classifier should not be able to discriminate between features; that is, it should not perform better than a random guess. This would correspond to an area under the receiver operating characteristic curve (\mbox{ROC-AUC}) close or equal to 0.5. In the event of a distributional shift, the ROC-AUC would be significantly above 0.5. In this case, an analysis of feature importance can help to identify which features explain distributional shifts.

We implement adversarial validation via gradient boosting with decision trees \citep{Ke.2017}. We run the analysis over 100 different training and validation splits and consistently arrive at an out-of-sample ROC-AUC of 1.0. In other words, the classifier can perfectly discriminate between features originating from setting~\A and setting~\B, thereby adding further evidence of distributional shifts. We then compute the average feature importance based on the mean absolute feature attribution \citep{Lundberg.2020}. This allows us to identify the features that are the most important in explaining distributional shifts. The top five features associated with the largest distributional shift are listed in \Cref{fig:feature_importance}. The results suggest that a large portion of the distributional shifts can be explained by differences in the production conditions and material specifications (\ie, minimum design temperature, required insulation thickness, test pressure, $\ldots$).

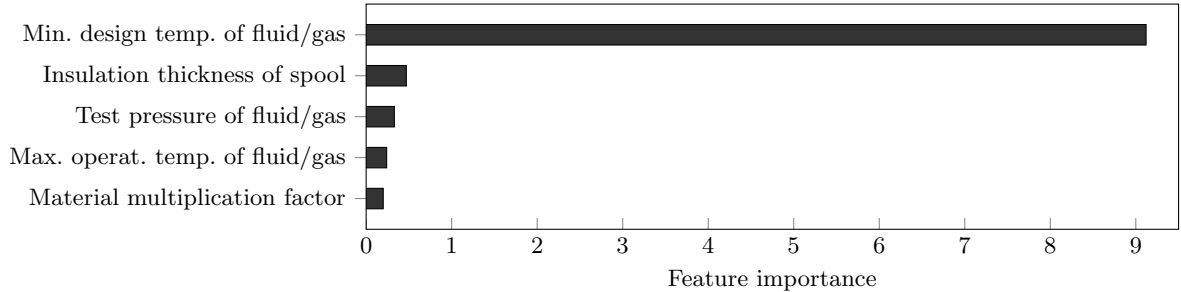
\begin{figure}
\centering
\footnotesize
\pgfplotstableread[header=false, col sep=comma]{ 
9.12, Min. design temp. of fluid/gas
0.47, Insulation thickness of spool
0.33, Test pressure of fluid/gas
0.24, Max. operat. temp. of fluid/gas 
0.20, Material multiplication factor
}\datatable

\begin{tikzpicture}
  \begin{axis}[
    ylabel near ticks,
    xlabel=Feature importance,
    xtick pos=bottom,
    ytick pos=left,
    xmin=0, 
    xmax=9.5,
    width=0.75\textwidth,
    xbar,
    bar width=2ex, y=4ex,
    enlarge y limits={abs=0.75},
    ytick={data}, 
    yticklabels from table={\datatable}{1}
  ]
  \addplot[black!20!black,fill=black!80!white] table [
    y expr=-\coordindex,
    x index=0 
  ] {\datatable};
  \end{axis}
\end{tikzpicture}
\caption{Features associated with the largest distributional shift.}
\label{fig:feature_importance}
\end{figure}

Overall, our exploratory analysis confirms that there are substantial distributional shifts between the two order settings. This can have a substantial effect on predictive models when training on data from setting~\A and predicting the operational outcomes (throughput times) for setting~\B. As we show in the following, the straightforward use of predictive analytics can harm prediction performance and thus scheduling decisions. 


\section[Model development]{Model development\footnote{Both codes and data are available via a public repository: \url{https://github.com/mkuzma96/CustomProd}.}}
\label{sec:methods}

In this section, we introduce our data-driven approach based on adversarial learning and job shop scheduling to support manufacturing operations in highly customized production. We first give a problem description (\Cref{sec:problem_statement}) and connect it with the concept of distributional shifts (\Cref{sec:problem_distributional_shifts}). Then, we adapt adversarial learning to address the distributional shifts in our decision problem (\Cref{sec:wdgrl}). 


\subsection{Problem description} 
\label{sec:problem_statement}

We consider a job shop production, where production orders are scheduled such that the expected costs of earliness (\ie, completing the order too early) and tardiness (\ie, completing the order too late) are minimized. We assume that each order $i=1, \ldots, m$ has a due date $d_i$ and a stochastic throughput time $y_i$. The actual throughput time $y_i$ is unknown a~priori and must thus be predicted from historical observations. We further assume that the to-be-completed production orders come from a forthcoming setting (called setting~\B) and we have access to labeled data from a historical setting (called setting~\A). Due to customized production, settings \A and \B are different. Formally, setting~\A involves $n$ individual production orders, where every order $i$ is described by order-specific features $x_i^\mathcal{A} \in \mathbb{R}^d$. Without loss of generality, \A can also be a set of multiple historical order settings. Further, let each individual order be associated with an observed order throughput time $y_i^\mathcal{A} \in \mathbb{R}$. The planned and forthcoming setting~\B comprises $m$ to-be-completed orders for which we would like to predict future throughput times. When starting to produce for order setting~\B, we have access to order-specific features $x_i^\mathcal{B} \in \mathbb{R}^d$ but not to the throughput times $y_i^\mathcal{B} \in \mathbb{R}$ because they lie ahead of time.

To optimize the job shop production, manufacturers typically follow a predict-then-optimize approach: in the first stage, the throughput times for customer orders are predicted, and, in the second stage, a job shop scheduling problem is solved in order to optimize scheduling decisions. We formalize this in the following.

\vspace{0.2cm}
\textbf{\emph{Stage~1.}} In the first stage, we estimate a predictive model $f: {X} \rightarrow {Y}$ to predict the throughput times $\hat{y}_i^\mathcal{B}$ for orders $i = 1, \ldots, m$ from the forthcoming setting~\B. Here, we can make use of order-specific features $x_i^\mathcal{B}$ that characterize the individual production orders (\eg, material specifications, process configurations, timeline information). The predicted throughput times for setting \B are given by $\hat{y}^\mathcal{B} = f(x^\mathcal{B})$. 

The input for estimating the predictive model $f$ is as follows. For the historical setting~\A, we have access to order-specific features $\{(x^\mathcal{A}_i)\}_{i=1}^n$ and the ground-truth throughput times $\{(y^\mathcal{A}_i)\}_{i=1}^n$. For the forthcoming setting~\B, we only have access to the order-specific features $\{(x^\mathcal{B}_i)\}_{i=1}^m$ but not the throughput times as these lie ahead of time. That is, we have a labeled dataset sampled from setting~\A and an unlabeled dataset sampled from the marginal distribution of setting~\B over $X$. Hence, the input is given by 
\begin{equation}
\{(x^\mathcal{A}_i, y^\mathcal{A}_i)\}_{i=1}^n \sim_{\text{iid}} (\mathcal{A})^n \qquad\text{and}\qquad \{(x^\mathcal{B}_i)\}_{i=1}^m \sim_{\text{iid}} (\mathcal{B}^X)^m .
\end{equation}

To estimate $f$, we aim at minimizing the expected error $\epsilon_{\mathcal{B}}$ between the predicted throughput times and observed throughput times for setting~\B. This yields the following objective
\begin{equation}
    \min_f \epsilon_\mathcal{B}(f) = \min_f \mathbb{E}_{(x^\mathcal{B},y^\mathcal{B})\sim \mathcal{B}}\left[\mathcal{L}(y^\mathcal{B},\hat{y}^\mathcal{B})\right] 
    \quad \text{given data } \{(x^\mathcal{A}_i, y^\mathcal{A}_i)\}_{i=1}^n \text{ and } \{(x^\mathcal{B}_i)\}_{i=1}^m 
    ,
    \label{eqn:target_error}
\end{equation}
where $\mathcal{L}$ denotes a convex loss function.\footnote{In this paper, we use the mean absolute error (MAE) for measuring prediction performance. This has the benefit that both prediction errors and scheduling costs are measured using the L1-norm; therefore, the costs in stages 1 and 2 are aligned. In principle, other loss functions such as the mean squared error (MSE) could also be used. However, using the MSE is not aligned with the cost penalty from the optimization stage. This is best seen in the following example: a predictive model with a small number of very large errors would have a large MSE but may have a small scheduling cost compared to a predictive with a large number of small errors.} Here, Eq.~(\ref{eqn:target_error}) is important: we aim at minimizing the expected error for setting~\B while having access to order-specific features for both setting~\A and setting~\B but only observed outcomes for setting~\A. We later adapt adversarial learning to make such predictions of throughput times while accounting for the heterogeneity between the order-specific features of setting~\A and setting~\B (see \Cref{sec:wdgrl}). In contrast to that, off-the-shelf machine learning would only minimize the expected error for setting \A without considering the distributional shift between the order-specific features of setting~\A and setting~\B.

\vspace{0.2cm}
\textbf{\emph{Stage~2.}} In the second stage, the scheduling task for the forthcoming setting~\B is formalized as an integer linear programming problem that determines a production sequence of the $m$ to-be-completed production orders given $T$ available time slots. Formally, let $\hat{y}_{i}$ denote the estimated order throughput time and $d_{i}$ the due date of a given production order $i$. The per time unit costs of earliness and tardiness are given by $c^{(early)}, c^{(tardy)} \in\mathbb{R}^{\geq 0}$, respectively. Further, let $K_{t}$ define the total number of production orders (\ie, capacity) that can be processed at time $t$. Then, the optimal production sequence can be solved via
\begin{subequations}
\begin{align}
\label{eqn:objective1}
\SingleSpacedXI
& 
\min\limits_{z_{it} \in \{0,1\}} \; \sum^{T}_{t=1} \sum^{m}_{i=1} z_{it} \, \Big[ c^{(early)}\, \max \big\{0,(d_{i}-t) - \hat{y}_{i} \big\} + c^{(tardy)}\, \max \big\{0,\hat{y}_{i}-(d_{i}-t) \big\}\Big]  \\
\text{s.t.} \quad & \qquad\qquad\qquad \psi(z, \hat{y}, t) \leq K_{t},  \qquad\text{for } t = 1,\ldots,T, \\
\label{eqn:objective3}
& \qquad\qquad\qquad \sum\limits_{t=1}^{T} z_{it} = 1, \qquad \;\;\;\;\;\;\;\; \text{for } i = 1,\ldots,m,
\end{align}
\end{subequations}
where the binary decision variable $z_{it}$ determines whether production order $i$ should be started at time $t$ and where the function $\psi(z, \hat{y}, t)$ counts the number of production orders that are produced in parallel. Here, the first constraint ensures that the available capacity is not exceeded, while the second constraint ensures that all production orders are fulfilled. We solve the optimization problem using the branch-and-cut implementation for mixed-integer problems from the GNU Linear Programming Kit (GLPK). Importantly, the scheduling task makes use of the throughput times $\hat{y}_i$, which are not given ex~ante but---analogous to OM practice---must be predicted a~priori before scheduling. 

The above problem is formulated as a predict-then-optimize approach due to three important practical benefits: (1)~It follows the current practice in order fulfillment. For example, a predict-then-optimize approach is consistent with decision-making at our case company Aker and other manufacturing firms. (2)~It offers great flexibility with regard to the chosen machine learning model. In particular, it allows manufacturers to use existing machine learning tools from their company. (3)~It allows manufacturers to incorporate expert knowledge. For example, manufacturers can assess the accuracy of the predictions before proceeding to the scheduling stage. 

Crucial to the above approach are accurate predictions of throughput times in the first stage. The reason is that incorrect predictions will lead to suboptimal production schedules and therefore additional costs. This can be formally seen in Eq.~(\ref{eqn:objective1}), where inaccurate predictions of throughput times  $\hat{y}_{i}$ negatively affect the overall production schedule. Hence, by accurately predicting throughput times (\emph{stage~1}), manufacturers can find optimal scheduling sequences (\emph{stage~2}), such that the cumulative costs of earliness and tardiness are minimized. However, predicting throughput times is particularly challenging when manufacturers produce highly customized products with non-standard specifications \citep{Cohen.2003}. Such customized products are characterized by large between-order heterogeneity and thus data samples that are not identically distributed, which, in turn, violates a standard assumption of machine learning \citep[cf.][]{Hastie.2009}. Hence, off-the-shelf machine learning models (\eg, standard deep neural networks) may give poor predictions of throughput times. The reason is that different specifications in customized production introduce distributional shifts between different customer order settings. This motivates an approach that accurately predicts order throughput times while accounting for distributional shifts between different order settings.

\subsection{Distributional shifts between order settings}
\label{sec:problem_distributional_shifts}

We now connect the heterogeneity among order settings to the concept of distributional shifts and thereby motivate the use of adversarial learning to give better predictions in stage~1. Recall that we consider different order settings where we denote the historical order setting by \A (in our empirical context, 5,830 component orders for a floating production vessel) and the forthcoming order setting by \B (in our empirical context, 3,866 component orders for an offshore oil platform). Here, we specifically focus on manufacturing settings with high degrees of product customization. In this particular context, forthcoming orders from new customers can involve entirely new specifications for which there are no historical observations. Varying specifications typically lead to between-order heterogeneity. Formally, the between-order heterogeneity is expressed by a distribution over $X\times Y$ that changes between model estimation (setting~\A) and model deployment (setting~\B). This is stated in the following definition of \emph{a distributional shift}.

\vspace{0.2cm}
\begin{definition}[Distributional shift] \emph{A distributional shift is a change in the joint probability distribution between the data from setting~\A that is used for model estimation and the data from setting~\B that is used during model deployment, \ie,}
\begin{equation}
\mathbb{P}_\mathcal{A}(X,Y) \neq \mathbb{P}_\mathcal{B}(X,Y).
\end{equation}
\end{definition}
\vspace{0.2cm}

\noindent
Building upon the concept of distributional shifts, we now explain why a na{\"i}ve application of off-the-shelf machine learning does not solve our task, and thereby we motivate the use of adversarial learning. In particular, we consider a specific case of distributional shift where the difference in the joint distribution  $\mathbb{P}(X,Y)$ between setting \A and setting \B results from a difference in the marginal distribution of $X$ (\ie, $\mathbb{P}(X)$), whereas the conditional distribution of $Y$ given $X$ (\ie, $\mathbb{P}(Y \mid X)$) remains unchanged between the two settings. Formally, we address a distributional shift of the type
\begin{equation}
\mathbb{P}_\mathcal{A}(X,Y) = \mathbb{P}(Y \mid X) \, \mathbb{P}_\mathcal{A}(X) \neq \mathbb{P}(Y \mid X) \, \mathbb{P}_\mathcal{B}(X) = \mathbb{P}_\mathcal{B}(X,Y),
\end{equation}
where $\mathbb{P}_\mathcal{A}(Y \mid X) = \mathbb{P}_\mathcal{B}(Y \mid X) = \mathbb{P}(Y \mid X)$, but $\mathbb{P}_\mathcal{A}(X) \neq \mathbb{P}_\mathcal{B}(X)$. This form of distributional shift is known in the literature as covariate shift \citep{Kouw.2018}. In fact, distributional shifts in form of covariate shifts are common at Aker and across OM practice. The latter essentially states that the specifications change between orders (\ie, $\mathbb{P}_\mathcal{A}(X) \neq \mathbb{P}_\mathcal{B}(X)$). For example, at Aker, one  setting may require thin and long spools while the other may require thick and short spools. The former essentially states that the process behind manufacturing products is comparable; that is, orders with identical specifications have the same throughput times regardless of whether they belong to setting \A or \B (\ie, $\mathbb{P}_\mathcal{A}(Y \mid X) = \mathbb{P}_\mathcal{B}(Y \mid X)$). For example, at Aker, thin (and long) spools will take the same time for production independent of whether the thin (and long) spool is later used in a floating production vessel or an offshore oil platform.

In a na{\"ive} application of machine learning, one would simply estimate $f$ only based on $(x^\mathcal{A},y^\mathcal{A})$. This has two key disadvantages (which later present two salient differences to our proposed approach). First, predictive models from off-the-shelf machine learning ignore the operational data from the forthcoming order, \ie, $x^{\mathcal{B}}$. However, such operational data characterizing forthcoming production orders are already available at the time of scheduling and could be used to improve the predictions and therefore the scheduling decisions. Second, predictive models from off-the-shelf machine learning optimize against $\min_f \mathbb{E}_{(x,y)\sim \mathcal{A}}\left[\mathcal{L}(y,\hat{y})\right] $ and not $\min_f \mathbb{E}_{(x,y)\sim \mathcal{B}}\left[\mathcal{L}(y,\hat{y})\right] $. That is, off-the-shelf machine learning optimizes the prediction performance for data coming from the historical probability distribution of data from setting~\A and not that of the forthcoming setting~\B. However, under a distributional shift $\mathbb{P}_\mathcal{A}(X,Y) \neq \mathbb{P}_\mathcal{B}(X,Y)$, both probability distributions are different; therefore, the optimization will not solve our objective from Eq.~(\ref{eqn:target_error}). The reason is that off-the-shelf machine learning makes the assumption of \emph{i.i.d.} sampling \citep[\cf][]{Hastie.2009} and thus that the distribution over $X\times Y$ remains unchanged between model estimation (setting~\A) and model deployment (setting~\B). This assumption does not hold in our manufacturing setting with high degrees of product customization. As a result, the performance of such off-the-shelf predictive models will deteriorate when deployed to a forthcoming setting~\B and will lead to suboptimal scheduling decisions.

\subsection{Proposed adversarial learning approach for predicting throughput times} 
\label{sec:wdgrl}

\subsubsection{Overview}

In the following, we address the objective from Eq.~(\ref{eqn:target_error}) through the use of adversarial learning. For this, we integrate the observed data $(x^\mathcal{A}$,$y^\mathcal{A})$ from setting~\A and the order-specific features $x^\mathcal{B}$ from setting~\B into the estimation of the predictive model $f$. This is referred to as an unsupervised domain adaptation problem. To provide a solution approach, we adapt adversarial learning to account for the different distributions behind \A and \B. Specifically, we take advantage of Wasserstein distance guided representation learning \citep[WDGRL;][]{Shen.2018}. WDGRL has been previously used for classification tasks in computer vision and computational linguistics but not for OM decisions. Using adversarial learning, we predict the throughput times while addressing the distributional shift and then solve the scheduling problem in Eq.~(\ref{eqn:objective1})--(\ref{eqn:objective3}) via integer optimization. Later, we confirm the effectiveness of adversarial learning over off-the-shelf machine learning for making job shop scheduling decisions under distributional shifts.

Formally, our aim is to predict throughput times under distributional shifts, that is, to achieve a low expected error in the forthcoming setting~\B, so that the scheduling decisions can be optimized. Because we have no observed outcomes for setting~\B, we cannot directly optimize the objective in Eq.~(\ref{eqn:target_error}). Nevertheless, the expected error in unsupervised domain adaptation problems can be bounded as stated in \Cref{remark:redko} \citep[adapted from][]{bendavid.2007, bendavid.2010, Redko.2017, Kouw.2018}. To state the remark, we first need a definition of the Wasserstein distance. 

\vspace{0.2cm}
\begin{definition}[Wasserstein distance]
\emph{The Wasserstein-1 (or Earth-Mover) distance between two probability distributions \A and \B is defined as,
\begin{equation}
    W(\mathcal{A},\mathcal{B}) = \inf_{\gamma \in \Pi(\mathcal{A},\mathcal{B})}\mathbb{E}_{(v,w)\sim\gamma}\left[\norm{v-w}\right],
\end{equation}
where $\Pi(\mathcal{A},\mathcal{B})$ is the set of all joint distributions $\gamma(v,w)$ with marginals \A and \B.}
\end{definition}
\vspace{0.2cm}

\noindent
Intuitively, the Wasserstein distance denotes the minimal amount of probability mass that must be transported (\eg, minimum expected transportation cost) from one distribution to the other to make them identical \citep{Arjovsky.2017}. 

\vspace{0.2cm}
\begin{remark}[\cite{Redko.2017}; Lemma 1]
\label{remark:redko}
The prediction error for setting~\B (\ie, $\epsilon_\mathcal{B}$) can be bound by the sum of the prediction error for setting~\A (\ie, $\epsilon_\mathcal{A}$) and the Wasserstein distance $W(\mathcal{A},\mathcal{B})$ between the feature distributions of settings \A and \B, \ie,
\begin{equation}
    \epsilon_\mathcal{B} \leq \epsilon_\mathcal{A} + W(\mathcal{A},\mathcal{B}) ,
\end{equation}
under some technical assumptions; see \citet{Redko.2017}.
\end{remark}
\vspace{0.2cm}

\noindent
The above remark assumes a machine learning classifier and provides the following theoretical motivation for our learning approach (see Supplement~\ref{appendix:bound} for a more detailed discussion of the error bound and the underlying technical assumptions). First, the upper bound of the prediction error depends on how well we can make predictions for setting~\A. Second, the upper bound of the prediction error should increase when the probability distributions of both settings drift apart. This motivates our adversarial learning approach where we aim to make inferences under two adversarial objectives: (1)~Our first objective is to estimate a function with a low prediction error on setting~\A. This is achieved by minimizing the loss between the actually observed outcomes and predictions from setting~\A. (2)~Our second objective accounts for the distance term $W(\mathcal{A}, \mathcal{B})$, whereby we learn latent feature representations of the order-specific features from both settings that are close to each other. More formally, we minimize the Wasserstein distance between the two feature distributions of setting~\A and setting~\B. As such, we aim for a good prediction performance in the known setting \A, but draw upon a representation that also generalizes well to operational data from the forthcoming setting~\B. The two aforementioned objectives are adversarial to each other (\eg, a close feature distribution does not imply a low error on setting~\A and vice versa). In the following, we formalize both objectives in a minimax game.

\subsubsection{Model specification}

Our adversarial learning approach is composed of three functions as follows (\Cref{fig:wdgrl_architecture}). (1)~A shared feature extractor $f_{e}: \mathbb{R}^d \rightarrow \mathbb{R}^l$ maps order-specific features from the $d$-dimensional input space $X$ of both order settings into a common latent space $\mathbb{R}^l$. This allows our approach to learn a shared representation of the latent feature distributions from both the historical and the forthcoming order setting. (2)~A regressor $f_{r}: \mathbb{R}^l \rightarrow \mathbb{R}$ outputs the prediction; that is, the throughput time given the latent features. (3)~A so-called critic $f_c: \mathbb{R}^l \rightarrow \mathbb{R}$ is used to estimate the Wasserstein distance between the latent feature distributions. We implement $f_e$, $f_r$, and $f_c$ as parameterized differentiable functions given by fully-connected linear feed-forward neural networks. Upon deployment, predictions are then made using $f_r \circ f_e$. That is, for input $x_i$, we compute the predicted throughput time via $\hat{y}_i = f_{r}(f_{e}(x_i))$.

The functions $f_e$, $f_r$, and $f_c$ are used in the two adversarial objectives as follows. The first adversarial objective ($\mathcal{L}_\text{reg}$) is to minimize expected prediction error for setting~\A and thus to learn predictions of the throughput time using data from the historical setting~\A. It involves $f_r \circ f_e$, which outputs the predictions. Formally, we can calculate the predicted outcome $\hat{y}_i$ via $f_{r}(f_{e}(x_i))$ for any $x_i$ sampled from setting~\A. The second adversarial objective ($\mathcal{L}_\text{was}$) aims to minimize the Wasserstein distance between settings \A and \B. Hence, it is based on $f_c \circ f_e$, so that the distance between the latent feature distributions of the historical and forthcoming setting is minimized.

\begin{figure}[h!]
\centering
\begin{subfigure}{0.5\textwidth}
  \centering
  \caption{\footnotesize{Model estimation}}
  \tikzset{every picture/.style={line width=0.75pt}} 

\begin{tikzpicture}[x=0.55pt,y=0.55pt,yscale=-1,xscale=1]
\footnotesize

\draw  [fill={rgb, 255:red, 155; green, 155; blue, 155 }  ,fill opacity=0.14 ] (0,33) .. controls (0,28.58) and (3.58,25) .. (8,25) -- (32,25) .. controls (36.42,25) and (40,28.58) .. (40,33) -- (40,117) .. controls (40,121.42) and (36.42,125) .. (32,125) -- (8,125) .. controls (3.58,125) and (0,121.42) .. (0,117) -- cycle ;
\draw  [fill={rgb, 255:red, 155; green, 155; blue, 155 }  ,fill opacity=0.14 ] (0,183) .. controls (0,178.58) and (3.58,175) .. (8,175) -- (32,175) .. controls (36.42,175) and (40,178.58) .. (40,183) -- (40,267) .. controls (40,271.42) and (36.42,275) .. (32,275) -- (8,275) .. controls (3.58,275) and (0,271.42) .. (0,267) -- cycle ;
\draw    (40,75) -- (67.5,75) ;
\draw [shift={(70.5,75)}, rotate = 180] [fill={rgb, 255:red, 0; green, 0; blue, 0 }  ][line width=0.08]  [draw opacity=0] (3.57,-1.72) -- (0,0) -- (3.57,1.72) -- cycle    ;
\draw  [fill={rgb, 255:red, 74; green, 144; blue, 226 }  ,fill opacity=0.4 ] (70,10) -- (85,10) -- (85,290) -- (70,290) -- cycle ;
\draw    (40,225) -- (67.5,225) ;
\draw [shift={(70.5,225)}, rotate = 180] [fill={rgb, 255:red, 0; green, 0; blue, 0 }  ][line width=0.08]  [draw opacity=0] (3.57,-1.72) -- (0,0) -- (3.57,1.72) -- cycle    ;
\draw    (85,75) -- (112.5,75) ;
\draw [shift={(115.5,75)}, rotate = 180] [fill={rgb, 255:red, 0; green, 0; blue, 0 }  ][line width=0.08]  [draw opacity=0] (3.57,-1.72) -- (0,0) -- (3.57,1.72) -- cycle    ;
\draw    (133.5,75) -- (161,75) ;
\draw [shift={(164,75)}, rotate = 180] [fill={rgb, 255:red, 0; green, 0; blue, 0 }  ][line width=0.08]  [draw opacity=0] (3.57,-1.72) -- (0,0) -- (3.57,1.72) -- cycle    ;
\draw  [fill={rgb, 255:red, 74; green, 144; blue, 226 }  ,fill opacity=0.4 ] (164,55) -- (179,55) -- (179,245) -- (164,245) -- cycle ;
\draw    (85,225.5) -- (112.5,225.5) ;
\draw [shift={(115.5,225.5)}, rotate = 180] [fill={rgb, 255:red, 0; green, 0; blue, 0 }  ][line width=0.08]  [draw opacity=0] (3.57,-1.72) -- (0,0) -- (3.57,1.72) -- cycle    ;
\draw    (133.5,225.5) -- (161,225.5) ;
\draw [shift={(164,225.5)}, rotate = 180] [fill={rgb, 255:red, 0; green, 0; blue, 0 }  ][line width=0.08]  [draw opacity=0] (3.57,-1.72) -- (0,0) -- (3.57,1.72) -- cycle    ;
\draw    (179,75) -- (226.33,75) ;
\draw [shift={(229.33,75)}, rotate = 180] [fill={rgb, 255:red, 0; green, 0; blue, 0 }  ][line width=0.08]  [draw opacity=0] (3.57,-1.72) -- (0,0) -- (3.57,1.72) -- cycle    ;
\draw  [fill={rgb, 255:red, 184; green, 233; blue, 134 }  ,fill opacity=0.55 ] (230,37.5) -- (245,37.5) -- (245,112.5) -- (230,112.5) -- cycle ;
\draw  [fill={rgb, 255:red, 226; green, 74; blue, 77 }  ,fill opacity=0.48 ] (230,187.5) -- (245,187.5) -- (245,262.5) -- (230,262.5) -- cycle ;
\draw    (179,75) -- (228.76,184.77) ;
\draw [shift={(230,187.5)}, rotate = 245.61] [fill={rgb, 255:red, 0; green, 0; blue, 0 }  ][line width=0.08]  [draw opacity=0] (3.57,-1.72) -- (0,0) -- (3.57,1.72) -- cycle    ;
\draw    (245,74.5) -- (272.5,74.5) ;
\draw [shift={(275.5,74.5)}, rotate = 180] [fill={rgb, 255:red, 0; green, 0; blue, 0 }  ][line width=0.08]  [draw opacity=0] (3.57,-1.72) -- (0,0) -- (3.57,1.72) -- cycle    ;
\draw    (293.5,74.5) -- (321,74.5) ;
\draw [shift={(324,74.5)}, rotate = 180] [fill={rgb, 255:red, 0; green, 0; blue, 0 }  ][line width=0.08]  [draw opacity=0] (3.57,-1.72) -- (0,0) -- (3.57,1.72) -- cycle    ;
\draw    (245,225) -- (272.5,225) ;
\draw [shift={(275.5,225)}, rotate = 180] [fill={rgb, 255:red, 0; green, 0; blue, 0 }  ][line width=0.08]  [draw opacity=0] (3.57,-1.72) -- (0,0) -- (3.57,1.72) -- cycle    ;
\draw    (293.5,225) -- (321,225) ;
\draw [shift={(324,225)}, rotate = 180] [fill={rgb, 255:red, 0; green, 0; blue, 0 }  ][line width=0.08]  [draw opacity=0] (3.57,-1.72) -- (0,0) -- (3.57,1.72) -- cycle    ;
\draw    (179,225.5) -- (226.33,225.5) ;
\draw [shift={(229.33,225.5)}, rotate = 180] [fill={rgb, 255:red, 0; green, 0; blue, 0 }  ][line width=0.08]  [draw opacity=0] (3.57,-1.72) -- (0,0) -- (3.57,1.72) -- cycle    ;
\draw  [fill={rgb, 255:red, 184; green, 233; blue, 134 }  ,fill opacity=0.55 ] (324.5,67) -- (339.5,67) -- (339.5,82) -- (324.5,82) -- cycle ;
\draw  [fill={rgb, 255:red, 226; green, 74; blue, 77 }  ,fill opacity=0.48 ] (324.5,217) -- (339.5,217) -- (339.5,232) -- (324.5,232) -- cycle ;
\draw  [dash pattern={on 0.84pt off 2.51pt}] (62.5,13.66) .. controls (62.5,6.12) and (68.62,0) .. (76.16,0) -- (173.84,0) .. controls (181.38,0) and (187.5,6.12) .. (187.5,13.66) -- (187.5,286.34) .. controls (187.5,293.88) and (181.38,300) .. (173.84,300) -- (76.16,300) .. controls (68.62,300) and (62.5,293.88) .. (62.5,286.34) -- cycle ;
\draw  [dash pattern={on 0.84pt off 2.51pt}] (221.5,39.84) .. controls (221.5,34.4) and (225.9,30) .. (231.34,30) -- (336.66,30) .. controls (342.1,30) and (346.5,34.4) .. (346.5,39.84) -- (346.5,110.16) .. controls (346.5,115.6) and (342.1,120) .. (336.66,120) -- (231.34,120) .. controls (225.9,120) and (221.5,115.6) .. (221.5,110.16) -- cycle ;
\draw  [dash pattern={on 0.84pt off 2.51pt}] (221.5,189.84) .. controls (221.5,184.4) and (225.9,180) .. (231.34,180) -- (336.66,180) .. controls (342.1,180) and (346.5,184.4) .. (346.5,189.84) -- (346.5,260.16) .. controls (346.5,265.6) and (342.1,270) .. (336.66,270) -- (231.34,270) .. controls (225.9,270) and (221.5,265.6) .. (221.5,260.16) -- cycle ;

\draw (1,75) node [text width=3cm, anchor=north, rotate=90][align=center]{\tiny{\baselineskip=10pt Data setting $\mathcal{A}$ \\ $(x^\mathcal{A},y^\mathcal{A}) \sim \mathcal{A}$ \par}};

\draw (1,225) node [text width=3cm, anchor=north, rotate=90][align=center]{\tiny{\baselineskip=10pt Data setting $\mathcal{B}$ \\ $(x^\mathcal{B}) \sim \mathcal{B}^X$ \par}};

\draw (350,75) node [text width=3cm, anchor=south, rotate=270][align=center]{\tiny{\baselineskip=7.5 pt Regression \\ loss $\mathcal{L}_{reg}$ \par}};

\draw (350,225) node [text width=3cm, anchor=south, rotate=270][align=center]{\tiny{\baselineskip=7.5 pt Wasserstein \\ loss $\mathcal{L}_{was}$ \par}};

\draw (125,-27.5) node [text width=3cm, anchor=north][align=center]{\baselineskip=10pt Feature extractor \par};

\draw (285,2.5) node [text width=3cm, anchor=north][align=center]{\baselineskip=10pt Regressor \par};

\draw (285,152.5) node [text width=3cm, anchor=north][align=center]{\baselineskip=10pt Critic \par};

\draw (124,77) node [anchor=center][inner sep=0.75pt]   [align=left] {\begin{minipage}[lt]{10.88pt}\setlength\topsep{0pt}
\begin{center}
...
\end{center}

\end{minipage}};
\draw (124,227.5) node [anchor=center][inner sep=0.75pt]   [align=left] {\begin{minipage}[lt]{10.88pt}\setlength\topsep{0pt}
\begin{center}
...
\end{center}

\end{minipage}};
\draw (284,76.5) node [anchor=center][inner sep=0.75pt]   [align=left] {\begin{minipage}[lt]{10.88pt}\setlength\topsep{0pt}
\begin{center}
...
\end{center}

\end{minipage}};
\draw (284,227) node [anchor=center][inner sep=0.75pt]   [align=left] {\begin{minipage}[lt]{10.88pt}\setlength\topsep{0pt}
\begin{center}
...
\end{center}

\end{minipage}};

\end{tikzpicture}
\end{subfigure}%
\begin{subfigure}{0.5\textwidth}
  \centering
  \caption{\footnotesize{Model deployment}}
  \tikzset{every picture/.style={line width=0.75pt}} 

\begin{tikzpicture}[x=0.55pt,y=0.55pt,yscale=-1,xscale=1]
\footnotesize

\draw  [fill={rgb, 255:red, 155; green, 155; blue, 155 }  ,fill opacity=0.14 ] (0,108) .. controls (0,103.58) and (3.58,100) .. (8,100) -- (32,100) .. controls (36.42,100) and (40,103.58) .. (40,108) -- (40,192) .. controls (40,196.42) and (36.42,200) .. (32,200) -- (8,200) .. controls (3.58,200) and (0,196.42) .. (0,192) -- cycle ;
\draw  [fill={rgb, 255:red, 74; green, 144; blue, 226 }  ,fill opacity=0.4 ] (70,10) -- (85,10) -- (85,290) -- (70,290) -- cycle ;
\draw    (40,150) -- (67.5,150) ;
\draw [shift={(70.5,150)}, rotate = 180] [fill={rgb, 255:red, 0; green, 0; blue, 0 }  ][line width=0.08]  [draw opacity=0] (3.57,-1.72) -- (0,0) -- (3.57,1.72) -- cycle    ;
\draw  [fill={rgb, 255:red, 74; green, 144; blue, 226 }  ,fill opacity=0.4 ] (164,55) -- (179,55) -- (179,245) -- (164,245) -- cycle ;
\draw    (85,150.5) -- (112.5,150.5) ;
\draw [shift={(115.5,150.5)}, rotate = 180] [fill={rgb, 255:red, 0; green, 0; blue, 0 }  ][line width=0.08]  [draw opacity=0] (3.57,-1.72) -- (0,0) -- (3.57,1.72) -- cycle    ;
\draw    (133.5,150.5) -- (161,150.5) ;
\draw [shift={(164,150.5)}, rotate = 180] [fill={rgb, 255:red, 0; green, 0; blue, 0 }  ][line width=0.08]  [draw opacity=0] (3.57,-1.72) -- (0,0) -- (3.57,1.72) -- cycle    ;
\draw    (179,151) -- (226.33,151) ;
\draw [shift={(229.33,151)}, rotate = 180] [fill={rgb, 255:red, 0; green, 0; blue, 0 }  ][line width=0.08]  [draw opacity=0] (3.57,-1.72) -- (0,0) -- (3.57,1.72) -- cycle    ;
\draw  [fill={rgb, 255:red, 184; green, 233; blue, 134 }  ,fill opacity=0.55 ] (230,113.5) -- (245,113.5) -- (245,188.5) -- (230,188.5) -- cycle ;
\draw    (245,150.5) -- (272.5,150.5) ;
\draw [shift={(275.5,150.5)}, rotate = 180] [fill={rgb, 255:red, 0; green, 0; blue, 0 }  ][line width=0.08]  [draw opacity=0] (3.57,-1.72) -- (0,0) -- (3.57,1.72) -- cycle    ;
\draw    (293.5,150.5) -- (321,150.5) ;
\draw [shift={(324,150.5)}, rotate = 180] [fill={rgb, 255:red, 0; green, 0; blue, 0 }  ][line width=0.08]  [draw opacity=0] (3.57,-1.72) -- (0,0) -- (3.57,1.72) -- cycle    ;
\draw  [fill={rgb, 255:red, 184; green, 233; blue, 134 }  ,fill opacity=0.55 ] (324.5,143) -- (339.5,143) -- (339.5,158) -- (324.5,158) -- cycle ;
\draw  [dash pattern={on 0.84pt off 2.51pt}] (62.5,13.66) .. controls (62.5,6.12) and (68.62,0) .. (76.16,0) -- (173.84,0) .. controls (181.38,0) and (187.5,6.12) .. (187.5,13.66) -- (187.5,286.34) .. controls (187.5,293.88) and (181.38,300) .. (173.84,300) -- (76.16,300) .. controls (68.62,300) and (62.5,293.88) .. (62.5,286.34) -- cycle ;
\draw  [dash pattern={on 0.84pt off 2.51pt}] (221.5,115.84) .. controls (221.5,110.4) and (225.9,106) .. (231.34,106) -- (336.66,106) .. controls (342.1,106) and (346.5,110.4) .. (346.5,115.84) -- (346.5,186.16) .. controls (346.5,191.6) and (342.1,196) .. (336.66,196) -- (231.34,196) .. controls (225.9,196) and (221.5,191.6) .. (221.5,186.16) -- cycle ;

\draw (1,150) node [text width=3cm, anchor=north, rotate=90][align=center]{\tiny{\baselineskip=10pt Data setting $\mathcal{B}$ \\ $(x^\mathcal{B}) \sim \mathcal{B}^X$ \par}};

\draw (350,150) node [text width=3cm, anchor=south, rotate=270][align=center]{\tiny{\baselineskip=7.5 pt Prediction \par}};

\draw (125,-27.5) node [text width=3cm, anchor=north][align=center]{\baselineskip=10pt Feature extractor \par};

\draw (285,80) node [text width=3cm, anchor=north][align=center]{\baselineskip=10pt Regressor \par};

\draw (284,152.5) node [anchor=center][inner sep=0.75pt]   [align=left] {\begin{minipage}[lt]{10.88pt}\setlength\topsep{0pt}
\begin{center}
...
\end{center}

\end{minipage}};
\draw (124,152.5) node [anchor=center][inner sep=0.75pt]   [align=left] {\begin{minipage}[lt]{10.88pt}\setlength\topsep{0pt}
\begin{center}
...
\end{center}

\end{minipage}};

\end{tikzpicture}
\end{subfigure}
\caption{Model specification based on the feature extractor $f_e$, the regressor $f_r$, and the critic $f_c$. The use of the model depends on whether (a) parameters should be estimated or (b) predictions should be made upon deployment. }  
 
\label{fig:wdgrl_architecture}
\end{figure}
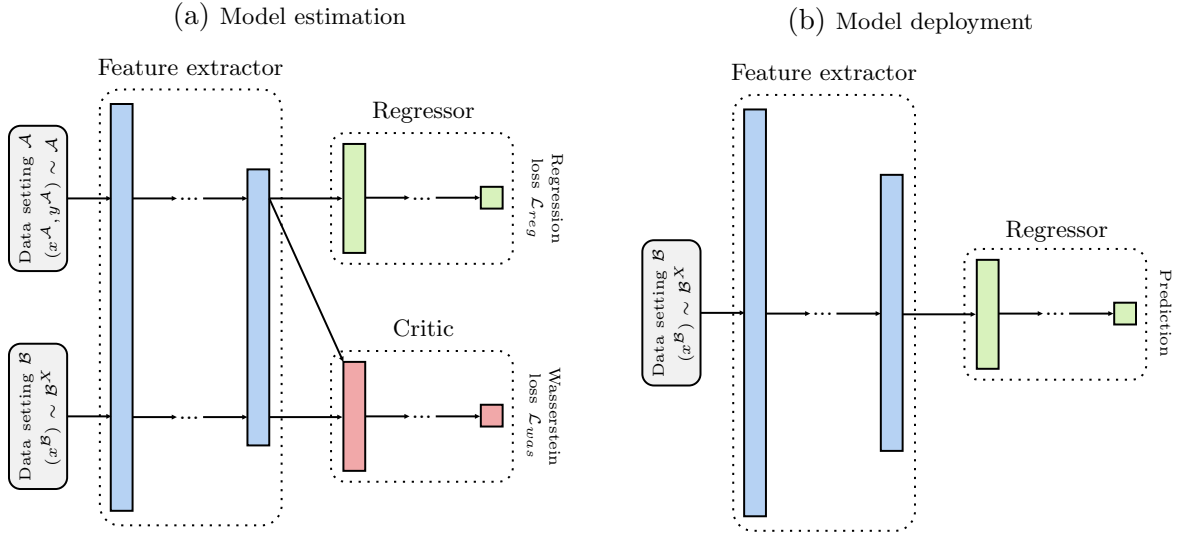

\vspace{0.2cm}
\noindent
\textbf{\emph{Adversarial objective 1 ($\mathcal{L}_\text{reg}$): Minimizing the expected prediction error for setting $\mathbf{\mathcal{A}}$}}
\vspace{0.1cm}

\noindent
The first objective is to achieve a low expected prediction error for setting~\A. For this, we compute the prediction error for labeled samples from setting~\A. The loss is given by
\begin{equation}
    \mathcal{L}_\text{reg} = \mathbb{E}_{(x,y)\sim \mathcal{A}} \,\,\mathcal{L}\left( f_{r}(f_{e}(x)),y\right) .
\end{equation}

\vspace{0.2cm}
\noindent
\textbf{\emph{Adversarial objective 2 ($\mathcal{L}_\text{was}$): Minimizing the Wasserstein distance between settings $\mathbf{\mathcal{A}}$ and $\mathbf{\mathcal{B}}$}}
\vspace{0.1cm}

\noindent
The second objective is to minimize the Wasserstein distance between settings \A and \B in the latent feature space. Recall that, intuitively, the Wasserstein distance denotes the minimal amount of probability mass that must be transported (\eg, minimum expected transportation cost) from one distribution to the other to make them identical \citep{Arjovsky.2017}. That is, we yield 
\begin{equation}
    W(\mathcal{A},\mathcal{B}) = \inf_{\gamma \in \Pi(\mathcal{A},\mathcal{B})}\mathbb{E}_{(v,w)\sim\gamma}\left[\norm{v-w}\right],
    \label{eqn:wdist_objective}
\end{equation}
where $\Pi(\mathcal{A},\mathcal{B})$ is the set of all joint distributions $\gamma(v,w)$ with marginals \A and \B. 

However, computing $W(\mathcal{A},\mathcal{B})$ directly is computationally infeasible. As a remedy, we make use of the critic function $f_c: \mathbb{R}^l \rightarrow \mathbb{R}$ \citep{Arjovsky.2017} by rewriting Eq.~(\ref{eqn:wdist_objective}) using the Kantorovich-Rubinstein duality, \ie,
\begin{equation}
    \label{eqn:wass_dist}
    \sup_{\norm{f_c}_L\leq 1} \mathbb{E}_{x \sim \mathcal{A}^X}\left[f_c(f_e(x))\right] - \mathbb{E}_{x \sim \mathcal{B}^X}\left[f_c(f_e(x))\right] ,
\end{equation}
where $f_c$ must be 1-Lipschitz. This is achieved by adding a gradient penalty loss $\mathcal{L}_\text{grad}$ \citep{Gulrajani.2017}, which penalizes the norm of the gradient itself. For this, we sample points $\hat{x}$ uniformly along straight lines between points sampled from \A and \B. The gradient penalty term can then be written as
\begin{equation}
    \mathcal{L}_\text{grad} =
    \mathbb{E}_{\hat{x}\sim\mathbb{P}_{\hat{x}}} \big[(\norm{\nabla_{\hat{x}}f_c(\hat{x})}_2-1)^2 \big].  
\end{equation}
Further details on the gradient penalty loss are provided in Supplement~\ref{appendix:model_details}.

Altogether, the loss to compute the Wasserstein distance is then given by 
\begin{equation}
\label{eqn:wloss}
    \mathcal{L}_\text{was} = \mathbb{E}_{x \sim A^X}\left[f_c(f_e(x))\right] - \mathbb{E}_{x \sim B^X}\left[f_c(f_e(x))\right]  - \beta\, \mathcal{L}_\text{grad}, 
\end{equation}
where $\beta$ is a gradient penalty weight. By maximizing the Wasserstein loss $\mathcal{L}_\text{was}$, we find the supremum in Eq.~(\ref{eqn:wass_dist}) that estimates the Wasserstein distance $W(\mathcal{A}, \mathcal{B})$ between the probability distributions of both settings. 

\subsubsection{Estimation procedure}

We now combine the two adversarial objectives given by the loss functions $\mathcal{L}_{\text{reg}}$ and $\mathcal{L}_{\text{was}}$ into a joint learning objective. Formally, this joint objective is given by the following minimax game
\begin{equation}
\label{eqn:obj}
    \min_{f_e,f_r} \Big\{ \mathcal{L}_\text{reg} + \alpha \, \max_{f_c} \mathcal{L}_\text{was} \Big\},
\end{equation} 
where $\alpha$ is a constant that weights the Wasserstein loss. Eq.~(\ref{eqn:obj}) aims at reducing the expected prediction error on the historical setting~\A through $\mathcal{L}_{\text{reg}}$, while simultaneously finding the supremum in Eq.~(\ref{eqn:wass_dist}) by maximizing $\mathcal{L}_\text{was}$ over $f_c$. The latter allows us to estimate the Wasserstein distance between the feature distributions and minimize it along with $\mathcal{L}_{\text{reg}}$ over $f_e$ and $f_r$. Note that the gradient penalty is used during maximization in order to estimate the Wasserstein distance, but the Wasserstein distance estimate itself does not contain the gradient penalty term (as seen in Eq.~(\ref{eqn:wass_dist})). Hence, the gradient penalty is not used during minimization; that is, for the $\min$ operation in Eq.~(\ref{eqn:obj}), we use the Wasserstein loss without the gradient penalty term.

In our implementation, we optimize the overall objective by alternating gradient descent following \cite{Shen.2018}. In every step, we first train the critic function to close optimality (according to the $\max$ operation in Eq.~(\ref{eqn:obj})), and then update the feature extractor and regressor by minimizing the regression loss, as well as the Wasserstein distance (\ie, the Wasserstein loss without gradient penalty) estimated by the critic. We further set the weights in the loss function, \ie, $\alpha$ from Eq.~(\ref{eqn:obj}) and $\beta$ from Eq.~(\ref{eqn:wloss}) to a default value of 1, so that we give equal weight to each part of the corresponding loss. We found that, regardless of the choice, the overall performance remains robust (see Supplement~\ref{appendix:robustness_adversarial_loss}).  

Upon model deployment, we only need the feature extractor and regressor (see \Cref{fig:wdgrl_architecture}) to make predictions. For order-specific features $x_i^\mathcal{B}$ from order setting~\B, we predict the throughput time via  
\begin{equation}
    \hat{y}^\mathcal{B}_i = f_{r}(f_{e}(x^\mathcal{B}_i)).
\end{equation}
Informed by \Cref{remark:redko}, this should then also minimize the prediction error $\epsilon_\mathcal{B}$ for the forthcoming order setting, and therefore address Eq.~(\ref{eqn:target_error}), that is, the objective of our prediction task in stage~1. Hence, our adversarial learning approach should achieve superior predictions in setting $\mathcal{B}$ compared to off-the-shelf machine learning methods that only focus on minimizing $\epsilon_\mathcal{A}$ (while ignoring operational data from setting $\mathcal{B}$). Afterward, we use the predictions to compute the corresponding optimal scheduling sequence by solving the integer optimization in Eq.~(\ref{eqn:objective1})--(\ref{eqn:objective3}), that is, the objective of our scheduling task in stage~2.

\subsection{Connecting prediction performance and scheduling costs}
\label{sec:theory}

We now provide arguments for why our proposed adversarial learning approach is effective for solving our decision problem. Previously, in \Cref{sec:problem_distributional_shifts} and \Cref{sec:wdgrl}, we motivated the use of adversarial learning to achieve better predictions compared to off-the-shelf machine learning in the presence of distributional shifts. Here, we discuss how the prediction performance translates to scheduling decisions and affects the scheduling costs. We first observe that the global minimum of scheduling cost is achieved for the oracle prediction in our decision problem, \ie, when the prediction error $\epsilon_\mathcal{B}$ is zero. Then, we use this argument to support the use of our adversarial learning approach that minimizes $\epsilon_\mathcal{B}$ over off-the-shelf machine learning which focuses only on minimizing $\epsilon_\mathcal{A}$.

When we discuss the scheduling costs in the following, we refer to the realized costs of earliness and tardiness (see \Cref{sec:job_shop_scheduling}) that result from a scheduling decision after the production is finished. For clarity, we provide an explicit definition below.
\begin{definition} 
\emph{The realized scheduling cost for a job shop scheduling problem is defined as
\begin{equation}
    \label{eqn:sch_cost}
    \bar{C}(\mathbf{z}) = \sum^{T}_{t=1} \sum^{m}_{i=1} z_{it} \, \Big[ c^{(early)}\, \max \big\{0,(d_{i}-t) - y_{i} \big\} + c^{(tardy)}\, \max \big\{0,y_{i}-(d_{i}-t) \big\}\Big],
\end{equation}
where $m$ is the number of production orders, $T$ is the number of available time slots, $z_{it} \in \{0,1\}$ are binary decision variables that determine at which time $t$ the production order $i$ started (\ie, $\sum^{T}_{t=1} z_{it} = 1$ for $i=1,\ldots,m$), $c^{(early)}$ and $ c^{(tardy)} \in\mathbb{R}^{\geq 0}$ are per time unit costs of earliness and tardiness, $d_i$ is a given due date of order $i$, and $y_i$ is the observed/realized throughput time of order $i$.}
\end{definition}

The realized scheduling cost depends on a scheduling decision $\mathbf{z} = (z_{11}, \ldots, z_{1T}, z_{21}, \ldots, z_{mT})$ which was made at time $t=0$, \ie, \emph{prior} to the beginning of the production process, whereas the realized scheduling cost is measured \emph{after} the production is finished and actual throughput times are observed. In the stage~2 of our approach, the scheduling decision $\mathbf{z}$ is calculated by solving the optimization problem in Eq.~(\ref{eqn:objective1})--(\ref{eqn:objective3}) given predictions for throughput times $\mathbf{\hat{y}} = (\hat{y}_1, \ldots, \hat{y}_m)$. Note that the realized scheduling cost function is of the same form as the cost function that we minimize in stage~2 of our approach (see Eq.~(\ref{eqn:objective1})). In fact, for the `oracle' predictions ($\mathbf{\hat{y}}_\text{oracle}$) that are equal to the true realized $\mathbf{y} $ (\ie when $\mathbf{\hat{y}} = \mathbf{\hat{y}}_\text{oracle} = \mathbf{y}$), the two cost functions are equivalent. Since a scheduling decision $\mathbf{z}$ is an argument that minimizes the cost function in Eq.~(\ref{eqn:objective1}), a decision that is calculated based on the `oracle' predictions $\mathbf{\hat{y}}_\text{oracle}$ is thereby a decision that minimizes the realized scheduling costs. In other words, the global minimum of the realized scheduling cost is achieved for a decision that is calculated based on the {`oracle'} prediction. We formalize this in Lemma~\ref{thm:min_cost} in Supplement~\ref{appendix:proof_thm}.

Therefore, perfectly accurate predictions of throughput times (with error $\epsilon_\mathcal{B}$ equal to zero) allow for optimal scheduling decisions that minimize the realized scheduling costs. However, the {`oracle'} prediction is a theoretical construct that is generally not available in OM practice. Rather, we have prediction algorithms with varying errors in prediction. Intuitively, the result in Lemma~\ref{thm:min_cost} suggests that, as we diverge from the `oracle' prediction and increase the prediction error (\ie, as we increase the MAE), the realized scheduling costs will increase correspondingly. In other words, the less accurate the predictions are (\ie, the larger the prediction MAE is), the larger the resulting scheduling costs will be. Hence, this underpins why our adversarial learning approach leads to superior decision-making compared to off-the-shelf machine learning: our approach aims to minimize $\epsilon_\mathcal{B}$, and therefore optimizes against the oracle predictions ($\mathbf{\hat{y}}_\text{oracle}$). In contrast, off-the-shelf machine learning minimizes $\epsilon_\mathcal{A}$, and therefore does \emph{not} optimize against the oracle predictions under a distributional shift. To sum up, our approach directly optimizes against the objective in Eq.~(\ref{eqn:target_error}) and thus achieves lower $\epsilon_\mathcal{B}$, which then translates into better scheduling decisions and ultimately cost savings.


\section{Numerical experiments} 
\label{sec:simulation_study}

In the following, we conduct a series of numerical experiments based on job shop scheduling to evaluate how distributional shifts in customized production affect production schedules, and therefore to study the operational value of our proposed adversarial learning approach. Further, to better understand the underlying mechanism of our approach, we further vary the operational setup across the following dimensions: (1)~the magnitude of the distributional shift, (2)~varying production line capacities, (3)~varying cost parameters, (4)~different distributions of the error term, and (5)~different nonlinearities in the operational data.

\subsection{Experimental setup}
\label{sec:exp_setup}

In the following, we set up a simulation where we vary the magnitude of distributional shifts. The simulation is designed to mimic decision-making in practice where the throughput times must be predicted before solving the job shop scheduling problem. 

We simulate data for the features $x_i^\mathcal{A}$ and $x_i^\mathcal{B}$ as follows. We first use the actual features from Aker from settings \A and \B to estimate the corresponding means, ${\mu}_{\mathcal{A}}$ and ${\mu}_{\mathcal{B}}$, and covariance matrices ${\Sigma}_{\mathcal{A}}$ and ${\Sigma}_{\mathcal{B}}$, respectively. For the historical setting \A, we then sample features from a multivariate Gaussian distribution with mean ${\mu}_{\mathcal{A}}$ and covariance matrix ${\Sigma}_{\mathcal{A}}$. This gives the samples $\{({x}^\mathcal{A}_i)\}_{i=1}^n$ with $n=5,830$ analogous to the dimension of the historical order setting $\mathcal{A}$ at Aker. For the forthcoming setting \B, we set up the sampling such that we can vary the magnitude of the distributional shift by introducing a mean shift in the direction of the difference between the means ${\mu}_{\mathcal{B}}$ and ${\mu}_{\mathcal{A}}$ (\ie, a mean shift from setting $\mathcal{A}$ towards setting $\mathcal{B}$). For this, we define the difference vector $v_{\text{diff}} = {\mu}_{\mathcal{B}} - {\mu}_{\mathcal{A}}$. Then, we sample features from a multivariate Gaussian distribution with mean ${\mu}_{\mathcal{A}} + \theta \cdot v_{\text{diff}}$ and covariance matrix ${\Sigma}_{\mathcal{B}}$, where parameter $\theta$ is used to control for the magnitude of the distributional shift (\ie, larger values for $\theta$ introduce larger distributional shifts). For a given $\theta$, this gives the samples $\{({x}^{\mathcal{B}}_i)\}_{i=1}^m$ with $m=3,866$ analogous to the dimensions of the forthcoming order setting $\mathcal{B}$ at Aker.

To simulate throughput times, we need a data-generating function $\phi : X \rightarrow Y$, so that we can generate throughput times conditional on some given features (this is needed since we have different features depending on the magnitude of the distributional shift in our simulation). We thus follow standard practice in machine learning  \citep[e.g.,][]{Shalit.2017, Yoon.2018} where so-called semi-synthetic datasets are used for modeling outcomes (in our case: throughput times). Specifically, we use predictive modeling to capture the data-generating process $\phi$ for Aker data in order to mimic the real-world setting at Aker. Here, the underlying choice of the machine learning model is crucial, because choosing some models may result in unfair advantages that bias later comparisons. For example, choosing a linear model for $\phi$ would strongly favor a linear model during evaluation. Similarly, choosing a neural network would favor our method because the structure of nonlinearities would be modeled in a similar way. Therefore, we use a nonlinear tree-based method, \ie, a random forest. By using structurally different nonlinearity to generate the data, we ensure that none of the methods has an unfair advantage later.\footnote{We later also perform a robustness check where we repeat our evaluation with a different $\phi$ and thus different nonlinearities. Specifically, we use gradient boosting but arrive at a similar conclusion: our adversarial learning approach remains consistently superior.} Formally, the throughput times are simulated by using the features via ${y} = {\phi}({x}) + \eta$, where $\phi$ is estimated using Aker data and where $\eta$ is Gaussian noise, \ie, $\eta \sim N(0, {\sigma}_{y})$, with ${\sigma}_{y}$ being the standard deviation estimate of $y$. For a given $\theta$, we thus obtain the simulated throughput times $\{({y}^\mathcal{A}_i)\}_{i=1}^n$, and  $\{({y}^\mathcal{B}_i)\}_{i=1}^m$, respectively. 

We use the following operational setup: (1)~The magnitude of the distributional shift is controlled by the parameter $\theta$. In all of our numerical experiments, we report results for $\theta = 1,2,3,4$ in order to examine how scheduling decisions are affected by different distributional shifts. (2)~The production line capacity in our main numerical experiment is set to $K_t=70$. We later also report results for $K_t = 50$ and thereby account for settings with smaller capacities such that production orders have to compete for production lines. (3)~We set the costs to $c^\text{early} = c^\text{tardy} = 1$ so that both earliness and tardiness are equally costly. We later also account for settings where overdue deliveries are more costly than finishing early (\ie, $c^\text{tardy} = 2$). (4)~We study how the distribution of the error term $\eta$ in our simulation affects scheduling costs. We thus change the Gaussian distribution to a uniform distribution with short tails. Formally, we use $\eta \sim \text{Unif}\big(-\frac{\sqrt{12} \hat{\sigma}_{y}}{2}, \frac{\sqrt{12} \hat{\sigma}_{y}}{2}\big)$ where the choice for the minimum and the maximum value is designed such that the standard deviation of $\eta$ remains equal to the standard deviation from before. (5)~We finally explore whether our results remain robust for varying nonlinearities in the operational data by changing the form of the data-generating process that we use to simulate the throughput times. Here, we repeat our numerical experiments where $\phi$ is given by gradient boosting with decision trees.  

Throughout this paper, we report (1)~the mean absolute error (MAE) for measuring prediction performance and (2)~the realized scheduling cost (as defined in Eq.~(\ref{eqn:sch_cost})) for measuring decision performance. The reason for choosing the MAE is that it allows us to measure errors in stage~1 using the L1-norm, which is thus aligned with stage~2. Note that we measure the \emph{out-of-sample} performance on setting~\B, that is, how well the approaches generalize to forthcoming order settings. Further, we account for variation in our simulation and thus report results from 10 different runs (\ie, we sample 10 different datasets using the above procedure). We later report the mean and the standard deviation. 

\subsection{Baselines}

We compare the following approaches for decision-making:
\begin{itemize}
\item \emph{Na{\"i}ve machine learning.} Here, we consider off-the-shelf machine learning methods that do not account for distributional shifts. We use two state-of-the-art machine learning methods for comparison: (1)~a \emph{regularized linear regression} (\ie, elastic net) as a linear method, and (2)~a \emph{deep neural network} as a state-of-the-art nonlinear method. Both methods are embedded in our predict-then-optimize framework and thus output scheduling decisions. This allows us to vary the prediction method in stage~1 of our decision problem, while the optimization for job shop scheduling remains identical across all methods. Hence, performance improvements must be attributed to that a method is better in addressing the objective for the predictions in stage~1.

\item \emph{Adversarial learning}. This is our proposed approach based on adversarial learning. Here, we use the same predict-then-optimize framework as for the off-the-shelf machine learning baselines. The only difference is that, through the use of adversarial learning, we now account for distributional shifts between order settings and thus directly address the objective of stage~1 in Eq.~(\ref{eqn:target_error}). For a fair comparison, we use the \emph{same} model architecture as for the deep neural network in our baselines. This is crucial: it rules out any performance gain due to the larger flexibility of the model. Instead, any performance gain must be solely attributed to the better learning objective. 
\item \emph{Oracle.} We report an oracle that has access to the ground-truth throughput times without noise. We then use the ground-truth throughput times when solving the optimization problem in Eq.~(\ref{eqn:objective1})--(\ref{eqn:objective3}). Note that the ground-truth throughput times are not available in practice. Instead, the purpose of the oracle is merely to offer a lower bound on the scheduling costs for comparability. 
\end{itemize}

\noindent

For all methods, we follow common practice \citep{Hastie.2009} and implement a rigorous hyperparameter search via 5-fold cross-validation on the data originating from setting~\A. Specifically, our hyperparameter search is an exhaustive search where all combinations are tested. Details on the hyperparameter tuning procedure can be found in Supplement~\ref{appendix:hyperparameter_tuning}. For the deep learning neural networks, we follow best-practice and use a multi-layer feed-forward neural network with ReLU activation and dropout regularization. We use the \emph{same} architecture in our adversarial learning approach. Importantly, we emphasize that our hyperparameter tuning is \emph{fair}. That is, both the deep neural network and our adversarial learning use exactly the \emph{same} tuning grid. Hence, the \emph{same} neural network architecture configurations are tested (\ie, implying that both have the {same} ``budget'' for tuning as both have similar runtimes). Hence, since all else is held equal, any improvements must solely be attributed to the fact that one of the two methods has a better objective function in stage~1. 

For the optimization, we use the branch-and-cut implementation for mixed-integer problems from the GNU Linear Programming Kit~(GLPK). All solver parameters are kept at their default values. Note that the optimization is identical for all of the above approaches. Due to computational reasons, we limit the optimization to the 100 production orders with the earliest due date.

\subsection{Results}

\subsubsection{Main result with a varying magnitude of the distributional shift}

We now report our main results (\Cref{tab:simulation_1}). Our results are reported as average values over 10 independent simulation runs ($\pm$ standard deviation). 

We make the following observations: (1)~Our adversarial learning approach consistently outperforms off-the-shelf machine learning baselines for different magnitudes of the distributional shift in terms of both prediction performance (MAE) and scheduling cost. We also performed a Welch's $t$-test that compares our approach to the better of the two off-the-shelf machine learning baselines. For $\theta = 2, 3, 4$, we find that the improvements in scheduling costs are statistically significant at the 0.1\%-significance threshold. (2)~The performance gains of our adversarial learning approach become larger when the magnitude of the distributional shift ($\theta$) is also larger. In other words, the performance of off-the-shelf machine learning quickly deteriorates as the magnitude of the distributional shift is increased, whereas our adversarial learning approach offers substantially more robust performance. For example, for $\theta = 4$, our adversarial learning approach achieves gains in the prediction performance of more than 24$\%$ over the baselines, which results in cost savings of more than 33\%. Hence, we observe that, across different magnitudes of the distributional shift, lower prediction errors of our adversarial learning lead to better scheduling decisions, ultimately resulting in lower realized scheduling costs. In sum, job shop scheduling using our adversarial learning approach is superior to job shop scheduling using off-the-shelf machine learning by a considerable margin.

\begin{table}
\renewcommand*{\arraystretch}{1.05}
\SingleSpacedXI
\footnotesize
\centering
\caption{Main results for job shop scheduling.} 
\label{tab:simulation_1}
\begin{tabular}{@{\extracolsep{4pt}} l cccc cccc}
\toprule
\textbf{Approach} & \multicolumn{4}{c}{Prediction error (MAE)} & \multicolumn{4}{c}{Scheduling cost} \\
\cline{2-5} \cline{6-9}
 & $\theta=1$ & $\theta=2$ & $\theta=3$ & $\theta=4$ & $\theta=1$ & $\theta=2$ & $\theta=3$ & $\theta=4$ \\
\midrule
Linear regression (regularized) & 24.3 & 25.3 & 28.1  & 31.5 & 2989.3 & 3355.0 & 3931.0 & 4453.9 \\
 & ($\pm$2.0) & ($\pm$1.6) & ($\pm$1.4) & ($\pm$2.1) & ($\pm$204.6) & ($\pm$250.5) & ($\pm$237.8) & ($\pm$270.9) \\
Deep neural network & 23.8 & 24.4 & 27.2 & 31.4 & 2933.1 & 3209.4 & 3870.9 & 4568.4 \\
 & ($\pm$2.0) & ($\pm$1.3) & ($\pm$1.5) & ($\pm$2.3) & ($\pm$216.8) & ($\pm$227.9) & ($\pm$310.6) & ($\pm$293.5) \\
Adversarial learning (ours) & 23.5 & 23.7 & 24.0 & 23.8 & 2782.2 & 2802.7 & 2995.2 & 2954.2 \\
 & ($\pm$1.7) & ($\pm$1.4) & ($\pm$2.1) & ($\pm$1.5) & ($\pm$182.3) & ($\pm$183.9) & ($\pm$177.7) & ($\pm$176.2) \\
\midrule
{Oracle (lower bound)} & 0.0 & 0.0 & 0.0 & 0.0 & 417.9 & 503.0 &  518.4 &  600.2 \\
 & --- & --- & --- & --- & ($\pm$58.8) & ($\pm$77.8) & ($\pm$73.9)  & ($\pm$74.7) \\
\bottomrule
\end{tabular}
\end{table}

\subsubsection{Sensitivity to different capacities and different costs}
\label{sec:c_params}

We now repeat our numerical experiments from above but vary the operational setup (see \Cref{fig:res}). First, we show the scheduling costs from the above numerical experiment for comparability (left). Second, we use a smaller production line capacity ($K_t = 50$), so that orders have to compete for production lines (center). Third, we use a different cost ratio where we account for the setting where overdue deliveries are more costly than finishing early (right). We thus set cost of tardiness to $c^\text{tardy} = 2$ and thus yield a cost ratio $\frac{c^\text{tardy}}{c^\text{early}} = 2$. All other parameters are identical to the above numerical experiments. The new numerical experiments have only an effect on the scheduling optimization in stage~2 while the predictions from stage~1 are identical to the previous numerical experiments. For that reason, we only report the scheduling costs for varying magnitudes of the distributional shift ($\theta = 1,2,3,4$).

\begin{figure}
\centering
\includegraphics[width=0.80\linewidth]{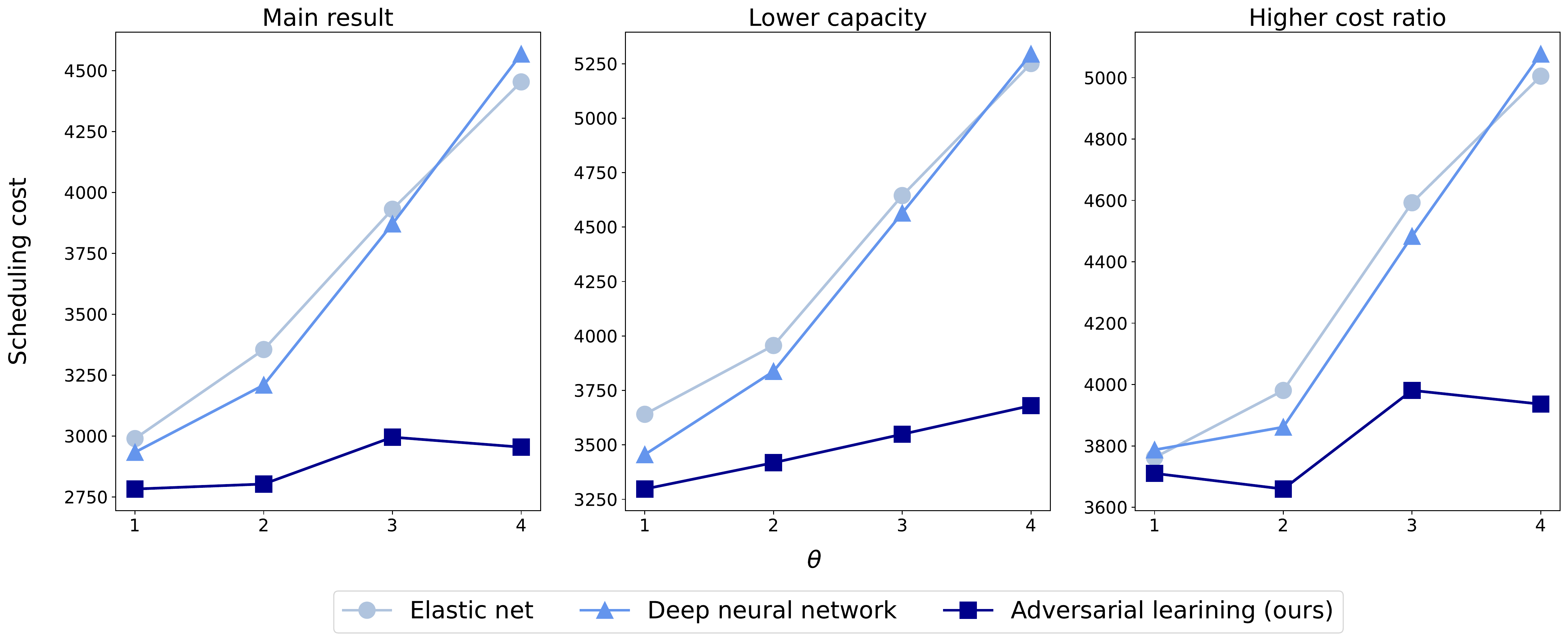}
\caption{Scheduling cost for increasing magnitude of distributional shift across different operational settings.}
\label{fig:res}
\end{figure}

Importantly, the main implications of our main numerical experiment remain unchanged (\Cref{fig:res}). (1)~Our adversarial learning approach outperforms off-the-shelf machine learning baselines. (2)~The performance gains from our adversarial learning approach increase when the magnitude of the distributional shift is large. We also note that the overall scheduling costs are larger for the numerical experiments with a lower capacity and a higher cost ratio in line with our expectations.

\subsubsection{Sensitivity to error distributions and nonlinearities}
\label{sec:error_dist}

We now repeat our main numerical experiment but vary how we generate data in our simulation. In \Cref{tab:error_dist}, we vary the distribution of the error term $\eta$ in our simulation. We now switch from a Gaussian (in our main numerical experiment) to a uniform distribution, which has shorter tails. The results confirm that our adversarial learning approach outperforms off-the-shelf machine learning baselines. We find this for different magnitudes of the distributional shift in terms of both prediction performance (MAE) and scheduling costs. In \Cref{tab:nonlin}, we repeat the simulation with a different nonlinearity in the operational data. Here, we replace $\phi$ with gradient boosting. As above, our adversarial learning approach consistently outperforms the off-the-shelf machine learning baselines for different magnitudes of the distributional shift. We further performed a  Welch's $t$-test that compares the scheduling costs from our approach to the better of the two off-the-shelf machine learning baselines. For $\theta = 2, 3, 4$, the improvements are again statistically significant at common significance thresholds.

Altogether, the results add to the robustness of our adversarial learning approach and demonstrate its operational value for job shop scheduling. Using real-world data from partner company Aker, we find consistent evidence that our adversarial learning approach leads to superior decision-making compared to off-the-shelf machine learning (\ie, current industry standard). It can therefore generate substantial cost savings.

\begin{table}
\renewcommand*{\arraystretch}{1.05}
\SingleSpacedXI
\footnotesize
\centering
\caption{Results for a different distribution of the error term.} 
\label{tab:error_dist}
\begin{tabular}{@{\extracolsep{4pt}} l cccc cccc}
\toprule
\textbf{Approach} & \multicolumn{4}{c}{Prediction error (MAE)} & \multicolumn{4}{c}{Scheduling cost} \\
\cline{2-5} \cline{6-9}
 & $\theta=1$ & $\theta=2$ & $\theta=3$ & $\theta=4$ & $\theta=1$ & $\theta=2$ & $\theta=3$ & $\theta=4$ \\
\midrule
{Linear regression (regularized)} & 25.3 & 27.6 & 29.3  & 33.0 & 2974.4 & 3612.1 & 3983.9 & 4700.2 \\
 & ($\pm$1.5) & ($\pm$1.6) & ($\pm$2.6) & ($\pm$2.7) & ($\pm$193.1) & ($\pm$223.0) & ($\pm$338.3) & ($\pm$349.2) \\
{Deep neural network} & 24.4 & 26.4 & 28.2 & 32.2 & 2891.2 & 3429.9 & 3865.4 & 4637.2 \\
 & ($\pm$1.5) & ($\pm$1.9) & ($\pm$2.9) & ($\pm$4.0) & ($\pm$165.7) & ($\pm$326.0) & ($\pm$422.1) & ($\pm$597.1) \\
{Adversarial learning (ours)} & 24.2 & 25.4 & 25.6 & 25.8 & 2741.0 & 3010.8 & 3020.3 & 3030.6 \\
 & ($\pm$1.3) & ($\pm$1.3) & ($\pm$1.4) & ($\pm$1.2) & ($\pm$154.7) & ($\pm$274.9) & ($\pm$250.7) & ($\pm$144.1) \\
\midrule
{Oracle (lower bound)} & 0.0 & 0.0 & 0.0 & 0.0 & 383.4 & 436.5 &  476.3 &  525.7 \\
 & --- & --- & --- & --- & ($\pm$82.5) & ($\pm$82.7) & ($\pm$85.2)  & ($\pm$72.2) \\
\bottomrule
\end{tabular}
\end{table}

\begin{table}
\renewcommand*{\arraystretch}{1.05}
\SingleSpacedXI
\footnotesize
\centering
\caption{Results for a different nonlinearity in the operational data.} 
\label{tab:nonlin}
\begin{tabular}{@{\extracolsep{4pt}} l cccc cccc}
\toprule
\textbf{Approach} & \multicolumn{4}{c}{Prediction error (MAE)} & \multicolumn{4}{c}{Scheduling cost} \\
\cline{2-5} \cline{6-9}
 & $\theta=1$ & $\theta=2$ & $\theta=3$ & $\theta=4$ & $\theta=1$ & $\theta=2$ & $\theta=3$ & $\theta=4$ \\
\midrule
{Linear regression (regularized)} & 23.9 & 24.6 & 27.1  & 30.0 & 2869.2 & 3183.2 & 3726.9 & 4065.2 \\
 & ($\pm$1.8) & ($\pm$1.3) & ($\pm$1.6) & ($\pm$1.1) & ($\pm$156.5) & ($\pm$249.3) & ($\pm$234.0) & ($\pm$192.5) \\
{Deep neural network} & 23.6 & 24.4 & 27.3 & 31.2 & 2823.9 & 3134.1 & 3737.7 & 4360.7 \\
 & ($\pm$1.9) & ($\pm$1.5) & ($\pm$2.0) & ($\pm$1.7) & ($\pm$187.7) & ($\pm$243.2) & ($\pm$283.0) & ($\pm$248.8) \\
{Adversarial learning (ours)} & 23.4 & 23.5 & 24.0 & 24.2 & 2697.5 & 2711.6 & 2815.0 & 2835.9 \\
 & ($\pm$1.7) & ($\pm$1.1) & ($\pm$1.7) & ($\pm$1.5) & ($\pm$141.8) & ($\pm$156.4) & ($\pm$162.4) & ($\pm$164.8) \\
\midrule
{Oracle (lower bound)} & 0.0 & 0.0 & 0.0 & 0.0 & 295.4 & 373.3 &  362.8 &  420.8 \\
 & --- & --- & --- & --- & ($\pm$54.2) & ($\pm$60.9) & ($\pm$58.6)  & ($\pm$55.8) \\
\bottomrule
\end{tabular}
\end{table}

\section{Robustness checks}
\label{sec:robustness_checks}

\subsection{Machine learning baselines for handling distributional shifts}
\label{sec:robustness_transfer_learning}

We consider other baselines from the machine learning literature \citep{Pan.2010, McNamara.2017, Kouw.2018} that can---in principle---also handle distributional shifts, namely model retraining and transfer learning with fine-tuning. However, we emphasize that the aforementioned baselines focus on a different setup called supervised domain adaptation (and not \underline{un}supervised domain adaptation, as in our decision problem). To this end, baselines from supervised domain adaptation require access to labels from setting \B and are thus only applicable \emph{after} the start of order setting~\B and not \emph{before}. This is a crucial difference to our adversarial learning approach, which is designed for operational contexts where labels for the forthcoming order setting~\B are absent (\ie, our approach has access to order-specific features $x_i^\mathcal{B}$ but \emph{not} to the corresponding labels $y_i^\mathcal{B}$.) In particular, for job shop scheduling, using model retraining and transfer learning would require access to some of the labels in setting \B, which means that one can perform scheduling optimization only \emph{after} production start and, hence, cannot provide an optimal scheduling sequence for all of the production orders in setting \B. 
Hence, we assume that the labels of the orders with throughput times that are within the first month of setting~\B are known, and, hence, scheduling optimization is done one month after the first spools have been produced. The results are shown in Supplement~\ref{appendix:robustness_ML_baselines}, where we see that these baselines are inferior, despite having access to more information than our approach.

\subsection{Baselines for domain adaptation}
\label{sec:robustness_adversarial_learning}

As an additional evaluation, we searched the literature for other domain adaptation baselines (see \cite{Wang.2018} for an overview). Here, another state-of-the-art baseline next to WDGRL is the so-called {gradient reversal layer} from \citet{Ganin.2016}. We find that WDGRL is more stable during training compared to a network architecture with a gradient reversal layer. This is consistent with earlier findings from the machine learning literature \citep{Shen.2018}.

\subsection{Robustness in other operational contexts}
\label{sec:robustness_reverse_order_setting}

As an additional robustness check, we consider a different operational context, that is, a different historical manufacturing project at Aker. Thereby, we show that our approach is transferable to other order settings. We again consider two order settings: Johan Sverdrup Living Quarters Rig is used for training in stage~1 of our approach, and Johan Castberg Floating Production Vessel is the order setting for which the job shop scheduling problem is solved. Both order settings do not have chronological overlap. The rest of the experimental setup is identical to \Cref{sec:exp_setup}. The results yield conclusive findings: job shop scheduling using an adversarial learning approach is superior over job shop scheduling using off-the-shelf machine learning by a considerable margin (see Supplement~\ref{appendix:robustness_other_settings}).


\section{Implications} 
\label{sec:implications}

\subsection{Methodological implications}

To meet order fulfillment targets, manufacturers typically follow a two-staged decision-making process where they first predict the throughput times of production orders and then determine an optimal production schedule. However, predicting throughput times in manufacturing settings with high degrees of product customization is challenging because of distributional shifts between customer orders. Such distributional shifts violate the standard assumption of identically distributed samples in predictive analytics \citep[cf.][]{Hastie.2009}, which can harm the prediction performance and thus lead to poor scheduling decisions. To account for distributional shifts, we developed a data-driven approach that combines adversarial learning and job shop scheduling. 

In our adversarial learning approach, we make predictions by modeling two adversarial objectives: (1)~to predict throughput times with the best possible prediction performance and (2)~to learn a neural network representation that generalizes well across order settings. Specifically, in the latter, we minimize the distance of the neural network representations of the operational data between the historical and the forthcoming order, which reduces prediction errors when applying the model to forthcoming orders with different specifications as the neural network representation is invariant to order settings. As such, the two adversarial objectives force predictions to not be biased towards historical orders but also account for the product specifications of forthcoming orders. This way, we capture distributional shifts and consequently improve decision-making in scheduling problems. 

While adversarial learning has almost exclusively been applied in computer vision and computational linguistics, this paper analyzes its operational value in an OM problem. As we have shown here, our adversarial learning approach can be effective for manufacturers that produce products with a high degree of customization, and, as such, it overcomes the limitations of existing methods in OM practice. In particular, our approach is different from conventional transfer learning \citep{McNamara.2017, Kouw.2018}, which requires access to labels of forthcoming orders (\ie, throughput times from the new customer order setting, yet these are only available after completion), whereas our approach circumvents the need for such labels.

\subsection{Managerial implications}

As the manufacturing industry is trending towards higher customization \citep{Feng.2018, Olsen.2020, Choi.2021}, decision models that can handle distributional shift gain relevance and importance. For example, Aker stressed the managerial implications of our work: their business is evolving towards more diverse products, higher volumes, and shorter lead times, which increase the distributional shifts and the relevance of our work for their operational decision-making. Hence, managers foresee a larger emphasis on addressing distributional shifts in the future. Yet, issues due to distributional shifts have received limited attention in OM research and practice. Our research contributes insights into how distributional shifts can be detected and provides operations managers with a promising approach to address them. Motivated by our work, we recommend practitioners to more carefully monitor operational data for potential distributional shifts (\eg, via adversarial validation as shown in \Cref{sec:analysis_of_distributional_shift}). For companies, this may serve as an early warning system to identify distributional shifts and inform managers when to take action.

Interestingly, our results suggest that using conventional machine learning---as common in OM practice---has an important limitation. It relies upon the assumption of identically distribution samples and, hence, cannot account for distributional shifts, which negatively affects prediction performance as well as scheduling costs. As such, OM practitioners must be aware that an off-the-shelf application of popular prediction models, such as deep neural networks, could result in poor decision-making. This has direct implications as conventional machine learning models are increasingly used for off-the-shelf predictions in OM practice \citep{Bastani.2021}. 

Our findings are also relevant beyond manufacturing. Distributional shifts are frequently observed in other areas of management \citep{Simester.2020}. In healthcare operations, for example, distributional shifts arise when predicting the mortality risk of rare or new diseases, or when applying machine learning to patient cohorts that are dissimilar from those upon training \citep{Hatt.2022}. In marketing, distributional shifts appear when making inferences about customer behavior in emerging segments. Likewise, distributional shifts may also arise for marginalized populations \citep{deArteaga.2022}.  Generally, adversarial learning has the potential to improve managerial decision-making in settings subject to extensive heterogeneity. 

\subsection{Limitations and opportunities for further research}

Our approach relies upon certain technical assumptions, which also hold for most unsupervised domain adaptation algorithms that rely on learning domain-invariant representations \citep{Kouw.2018}. First, we have assumed a specific form of distributional shift, namely a covariate shift \citep{Kouw.2018}, where the distribution of $X$ changes between the two settings $\mathcal{A}$ and $\mathcal{B}$, but the conditional distribution of $Y$ given $X$ remains constant. Covariate shifts imply that the manufacturing processes are comparable across products (that is, identical specifications lead to the same throughput time regardless of the underlying setting) and are thus common in OM practice. Second, another common assumption in unsupervised domain adaptation is that $P(X)$ has overlapping support between different domains. We have not explicitly made this assumption since the theoretical bounds that motivate our approach do not rely on overlapping support \citep{bendavid.2007, Redko.2017, Shen.2018}. However, many works in unsupervised domain adaptation stress the importance of both assumptions to guarantee successful learning (see \cite{Johansson.2019, Breitholz.2023} for more details). Overlapping support should also hold in OM practice as it implies that there is some similarity across orders. Still, decision-makers should be careful when using adversarial learning in cases where these assumptions do not hold. We call for further research into these limitations and the development of methods that are robust to them.


\subsection{Concluding remarks} 
\label{sec:conclusion}

In this paper, we showed that distributional shifts impair decision-making in operational settings. As a remedy, we proposed a data-driven approach combining adversarial learning and job shop scheduling where we address distributional shifts in customized production. Finally, we demonstrated its operational value using a series of numerical experiments based on a real-world job shop production at Aker. An important implication of our work for OM is that both practitioners and researchers need to be aware of potential risks due to distributional shifts in operational settings and, if these occur, must seek effective ways to address them.


\vspace{0.3cm}
\noindent
{\footnotesize\textbf{\textsf{Code and data availability.}} Both codes and data are available via a public repository: \url{https://github.com/mkuzma96/CustomProd}.}


\ACKNOWLEDGMENT{We thank Trond Haga, Sigmund Mongstad Hope, and Jürg Käser for our cooperation with Aker Solutions. We also acknowledge the constructive comments we received from the editors and reviewers. This research received financial support from the Norwegian Research Council (309810 COM-FLEX: Competitive Flexibility).}

\bibliographystyle{pomsref} 

 \let\oldbibliography\thebibliography
 \renewcommand{\thebibliography}[1]{%
 	\oldbibliography{#1}%
 	\baselineskip12.5pt 
 	\setlength{\itemsep}{3pt}
 }

 \SingleSpacedXI
\bibliography{references}
\OneAndAHalfSpacedXI

\clearpage

\newcommand{\hbAppendixPrefix}{A}
\renewcommand{\thefigure}{\hbAppendixPrefix\arabic{figure}}
\setcounter{figure}{0}
\renewcommand{\thetable}{\hbAppendixPrefix\arabic{table}} 
\setcounter{table}{0}
\renewcommand{\theequation}{\hbAppendixPrefix\arabic{equation}} 
\setcounter{equation}{0}
\def\thesection{\Alph{section}}
\setcounter{section}{0}

~~~~~
\vfill
\begin{center}
\Large
\textbf{Online Supplements}
\normalsize
\end{center}
\vfill


\clearpage 
\section{Theoretical bounds for adversarial learning}
\label{appendix:bound}

In \Cref{sec:wdgrl}, we discussed theoretical bounds for unsupervised domain adaptation that motivate the use of domain-invariant representation learning to minimize the distances between the feature distributions. Specifically, we rely on the Wasserstein distance metric, which we summarize in \Cref{remark:redko} (adapted from \cite{Redko.2017}). Consequently, our approach uses Wasserstein distance guided representation learning \citep{Shen.2018}. However, we acknowledge that the \Cref{remark:redko} relies on certain technical assumptions (see \cite{Redko.2017}), which may differ from our setting. We primarily refer to the fact that the theoretical bounds that motivate our approach are derived for a classification task with binary labels \citep{bendavid.2007, Redko.2017, Shen.2018}. However, \citet{Mansour.2009} have shown that such theoretical bounds can be generalized for domain adaptation problems with arbitrary loss functions, and hence domain-invariant representation learning approaches are applicable for regression tasks as well. This provides additional theoretical backing for using the proposed approach in our case of predicting throughput times in customized prediction.

Additionally, we discuss another bound based on the Wasserstein distance used by Wasserstein distance guided representation learning \citep{Shen.2018}, and provide a formal analysis showing that the proposed adversarial learning architecture provides a solution approach to our objective in Eq.~(\ref{eqn:target_error}). For this, we give an upper bound for the expected prediction error when making predictions for order setting~\B. To derive a bound, we need to adapt our setting by assuming to make predictions on whether $y$ (the throughput time measured in time units) exceeds a certain threshold $\theta$ or not. Provided the labels are written as $y' = \mathbbm{1}_{y\geq\theta}$, then the following theorem holds.

\begin{theorem}[Theorem 1 in \cite{Shen.2018}]
Let setting~\A and setting~\B have the same labeling function $f^*: X \rightarrow [0,1]$ and let $\mathcal{H}$ be a hypothesis class, that is, the set of predictor functions. For all $\forall h\in \mathcal{H}, h: X\rightarrow [0,1]$ that are $K$-Lipschitz continuous for some $K$, the prediction error $\epsilon_\mathcal{B}$ is bounded by
\begin{equation*}
    \epsilon_\mathcal{B}(h) \leq \epsilon_\mathcal{A}(h) + 2 \,\,K \,\, W_1(\mathcal{A}, \mathcal{B}) + \lambda,
\end{equation*}
where $W_1$ is the Wasserstein distance and $\lambda$ is the combined error of the ideal hypothesis $h^*$ that minimizes both $\epsilon_\mathcal{B}(h)+\epsilon_\mathcal{A}(h)$.
\end{theorem}

\begin{proof}{}
See Theorem 1 in \cite{Shen.2018}.
\end{proof}

\noindent
While this bound further supports our proposed adversarial learning approach from a theoretical point of view, we emphasize that the inequality is probabilistic and the error term $\lambda$ can be large in practice. Such bounds are common practice in unsupervised domain adaptation where the true outcomes from setting~\B are not available during training \citep{bendavid.2007, Redko.2017, Shen.2018}. Therefore, in this work, we additionally validate our approach empirically using an extensive series of numerical experiments.

\clearpage 
\section{Background on gradient penalty loss} 
\label{appendix:model_details}

In Eq.~(\ref{eqn:wloss}), we use a gradient penalty loss $\mathcal{L}_{\text{grad}}$ to enforce the 1-Lipschitz constraint. Originally, \cite{Arjovsky.2017} use clipping of the neural network weights to constrain the gradients of the critic function, which has proven to be unstable during training. Therefore, we use the penalty term introduced by \cite{Gulrajani.2017}, which penalizes the norm of the gradient itself. Recall the gradient norm should be at most 1 in order for the critic to be 1-Lipschitz. \citet{Gulrajani.2017} sample points $\hat{x}$ uniformly along straight lines between points sampled from \A and \B. The gradient penalty term can then be written as
\begin{equation}
    \mathcal{L}_\text{grad} =
    \mathbb{E}_{\hat{x}\sim\mathbb{P}_{\hat{x}}} \big[(\norm{\nabla_{\hat{x}}f_c(\hat{x})}_2-1)^2 \big].  
\end{equation}
For further theoretical details, we refer to \citet{Arjovsky.2017} and \citet{Gulrajani.2017}. 

\clearpage 
\section{Global minimum of the realized scheduling cost}
\label{appendix:proof_thm}

\begin{lemma}
\label{thm:min_cost}
Let $\bar{C}(\mathbf{z})$ be the realized scheduling cost as in Eq.~(\ref{eqn:sch_cost}) for a given scheduling decision $\mathbf{z}$, and let a scheduling decision $\mathbf{z}(\mathbf{\hat{y}})$ be the solution to the optimization problem in Eq.~(\ref{eqn:objective1})-(\ref{eqn:objective3}) for given predictions of throughput times $\mathbf{\hat{y}} = (\hat{y}_1, \ldots, \hat{y}_m)$. Further, let $\mathbf{\hat{y}}_\text{oracle} = (y_1, \ldots, y_m)$ be the `oracle' prediction that is equal to the realized future throughput times. Then, for every $\hat{\mathbf{y}} \in \mathbb{R}^m$, we have
\begin{equation}
\label{thm:result}
    \bar{C}\big(\mathbf{z}(\hat{\mathbf{y}}_\text{oracle})\big) \leq \bar{C}\big(\mathbf{z}(\hat{\mathbf{y}})\big).
\end{equation}
\end{lemma}

\noindent

\begin{proof}{}
By definition of the scheduling task in stage~2 (Eq.~(\ref{eqn:objective1})-(\ref{eqn:objective3})), we have
\begin{subequations}
\begin{align}
\SingleSpacedXI
& 
\mathbf{z}(\mathbf{\hat{y}}) = \underset{{z_{it} \in \{0,1\}}}{\arg\min} \; \sum^{T}_{t=1} \sum^{m}_{i=1} z_{it} \, \Big[ c^{(early)}\, \max \big\{0,(d_{i}-t) - \hat{y}_{i} \big\} + c^{(tardy)}\, \max \big\{0,\hat{y}_{i}-(d_{i}-t) \big\}\Big]  \nonumber \\
\text{s.t.} \quad & \qquad\qquad\qquad \psi(z, \hat{y}, t) \leq K_{t},  \qquad\text{for } t = 1,\ldots,T, \nonumber \\
& \qquad\qquad\qquad \sum\limits_{t=1}^{T} z_{it} = 1, \qquad \;\;\;\;\;\;\;\; \text{for } i = 1,\ldots,m. \nonumber
\end{align}
\end{subequations}
Further, by the definition of the \emph{`oracle'} prediction, $\mathbf{\hat{y}}_{oracle}$, we have 
\begin{subequations}
\begin{align}
\SingleSpacedXI
& 
\mathbf{z}(\mathbf{\hat{y}}_{oracle}) = \underset{{z_{it} \in \{0,1\}}}{\arg\min} \; \sum^{T}_{t=1} \sum^{m}_{i=1} z_{it} \, \Big[ c^{(early)}\, \max \big\{0,(d_{i}-t) - y_{i} \big\} + c^{(tardy)}\, \max \big\{0,y_{i}-(d_{i}-t) \big\}\Big]  \nonumber \\
\text{s.t.} \quad & \qquad\qquad\qquad \psi(z, \hat{y}, t) \leq K_{t},  \qquad\text{for } t = 1,\ldots,T, \nonumber \\
& \qquad\qquad\qquad \sum\limits_{t=1}^{T} z_{it} = 1, \qquad \;\;\;\;\;\;\;\; \text{for } i = 1,\ldots,m. \nonumber
\end{align}
\end{subequations}
Then, for the realized scheduling cost given in Eq.~(\ref{eqn:sch_cost}), we have
\begin{subequations}
\begin{align}
\SingleSpacedXI
& 
\mathbf{z}(\mathbf{\hat{y}}_{oracle}) = \underset{{z_{it} \in \{0,1\}}}{\arg\min} \; \bar{C}(\mathbf{z})  \nonumber \\
\text{s.t.} \quad & \qquad\qquad\qquad \psi(z, \hat{y}, t) \leq K_{t},  \qquad\text{for } t = 1,\ldots,T, \nonumber \\
& \qquad\qquad\qquad \sum\limits_{t=1}^{T} z_{it} = 1, \qquad \;\;\;\;\;\;\;\; \text{for } i = 1,\ldots,m. \nonumber
\end{align}
\end{subequations}
Finally, since $\mathbf{z}(\mathbf{\hat{y}}_{oracle})$ is the argument that minimizes $\bar{C}(\mathbf{z})$, the result in Lemma~\ref{thm:min_cost}, Eq.~(\ref{thm:result}), follows directly for every $\mathbf{z}(\mathbf{\hat{y}}) \in \{0,1\}^m$, \ie, for every $\hat{\mathbf{y}} \in \mathbb{R}^m$.
\end{proof}

\clearpage 
\section{Implementation details (hyperparameter tuning)}
\label{appendix:hyperparameter_tuning}

\emph{Preprocessing:} Following best practice in predictive analytics \citep{Hastie.2009}, we standardized all features, that is, we transformed all features such that they have zero mean and unit variance. We also experimented with other transformations (\eg, transforming all features to the range $[-1,1]$) but without improvement. 

\emph{Implementation:} The adversarial learning model as well as all machine learning baselines and models used to generate throughput times are implemented using Python. All computations are executed on conventional office hardware (Intel Core i7, 1.8 GHz) within 20 minutes (\ie for fitting all of the models on one data sample, plus scheduling optimization afterwards). The adversarial learning approach and the deep neural network baseline are implemented using \emph{tensorflow} and \emph{keras}. We implemented the regularized linear regression (elastic net) baseline using \emph{scikit-learn}. For estimating the data generating function $\phi: X \rightarrow Y$ from the real-world data in order to simulate throughput times using the random forest and gradient boosting algorithm, we use \emph{scikit-learn} and the \emph{LightGBM} package, respectively. 

\emph{Hyperparameter tuning:} \Cref{tab:grid_search} lists the tuning ranges for all hyperparameters. Tuning was done using grid search with 5-fold cross-validation. Note that grid search is an exhaustive search in which all combinations are tested. We emphasize that our tuning is \emph{\textbf{fair}}:  for our adversarial learning approach and the deep neural network, we have the \emph{same} grid for tuning the hyperparameters. In other words, exactly the \emph{same} hyperparameters are tested during the exhaustive search, and both have thus exactly the same budget for tuning. Hence, there is \emph{no} advantage of one method over the other from better hyperparameter tuning. 

In our implementation, we tuned our adversarial learning approach as follows. Recall that, during training, we can only evaluate the prediction performance on \A, and not yet on \B due to the absence of data. Hence, for reasons of fairness, we simply use the same hyperparameters for our adversarial learning approach as for the deep neural network baseline. This is again beneficial for benchmarking, as the architecture of our adversarial learning approach is truly \emph{identical} to the deep neural network baseline. This means: differences in the performance cannot be attributed to different hyperparameters (or different tuning). Instead, they are the result of the different learning objectives. 

\clearpage

\begin{table}[H]
\begin{center}
\caption{Grid search for hyperparameter tuning}  \label{tab:grid_search}
\scriptsize
\renewcommand*{\arraystretch}{1.1}
\begin{tabular}{lll}
\toprule
\textbf{Model} & \textbf{Tuning parameters}  & \textbf{Tuning range}\\
\midrule
Elastic net & Regularization strength & 0.01, {0.1}, 1, 10, 100\\
& Regularization ratio & 0, 0.25, 0.5, 0.75, {1}\\[0.3em]
Deep neural network & Learning rate & {0.0005}\\
& Number of layers in feature extractor & {1}\\
& Number of neurons in feature extractor & 8, 16, 32, {64}\\
& Number of layers in regressor & {1}\\
& Number of neurons in regressor & {8}\\
& Epochs & {50, 100}\\
& Batch size & {32, 64}\\
& Kernel regularization strength & {0.00001}, {0.001}, 0.01\\
& Dropout rate & {0.4}, 0.5, 0.6\\[0.3em]
Adversarial learning approach (ours) & Learning rate & {0.0005}\\
& Number of layers in feature extractor & {1}\\
& Number of neurons in feature extractor & 8, 16, 32, {64}\\
& Number of layers in regressor & {1}\\
& Number of neurons in regressor & {8}\\
& Epochs & {50, 100}\\
& Batch size & {32, 64}\\
& Kernel regularization strength & {0.00001}, {0.001}, 0.01\\
& Dropout rate & {0.4}, 0.5, 0.6\\[0.3em]
\bottomrule
\end{tabular}
\end{center}
\smallskip
\end{table}

\clearpage

\section{Robustness checks for adversarial loss}
\label{appendix:robustness_adversarial_loss}

Our adversarial learning approach has two weights, namely $\alpha$ from Eq.~(\ref{eqn:obj}) and $\beta$ from Eq.~(\ref{eqn:wloss}). Importantly, these weights  are part of the loss function and not of the network architecture. Therefore, these weights do not lead to a more flexible network architecture, but simply control how the adversarial objectives are learned. In our main paper, we set both weights to a default value of 1, which implies that the different parts of the loss receive equal weight. 

Here, we conduct a sensitivity analysis to confirm that, regardless of the choice, the results remain robust. We use the experimental setup from \Cref{sec:exp_setup} and vary $\alpha, \beta = 0.8, 0.9, 1.0, 1.1, 1.2$ (while keeping the other one fixed to default value 1.0). We again report separate results for different domain shifts $\theta=1,2,3,4$. To make direct comparisons, we only report the deep neural network baseline as it has the same network architecture as our adversarial learning approach. Furthermore, to enable direct comparability along $\alpha$ and $\beta$ (and not along the overall architecture), we use the same hyperparameters from the main experiment, which are the result of a rigorous hyperparameter tuning. The results are shown in \Cref{fig:vary_alpha} for weight $\alpha$ and in \Cref{fig:vary_beta} for weight $\beta$. The deep neural network baseline (without adversarial learning objective) is outperformed by a considerable margin, regardless of how the weights are initialized. This confirms that the performance of our adversarial learning approach remains robust.

\begin{figure}
\centering
\includegraphics[width=0.7\linewidth, height=20cm]{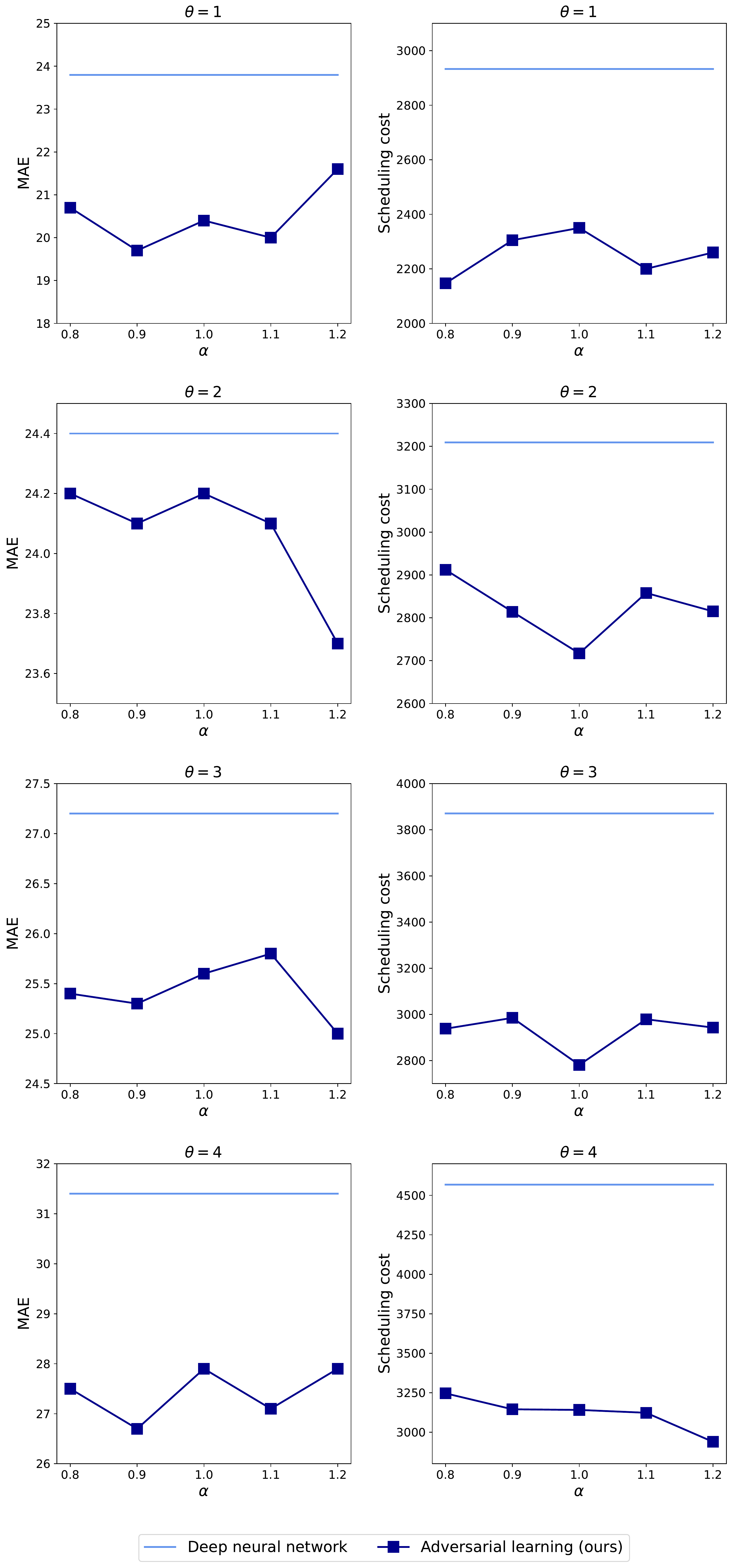}
\caption{Robustness checks for adversarial loss with respect to weight $\alpha$.}
\label{fig:vary_alpha}
\end{figure}

\begin{figure}
\centering
\includegraphics[width=0.7\linewidth, height=20cm]{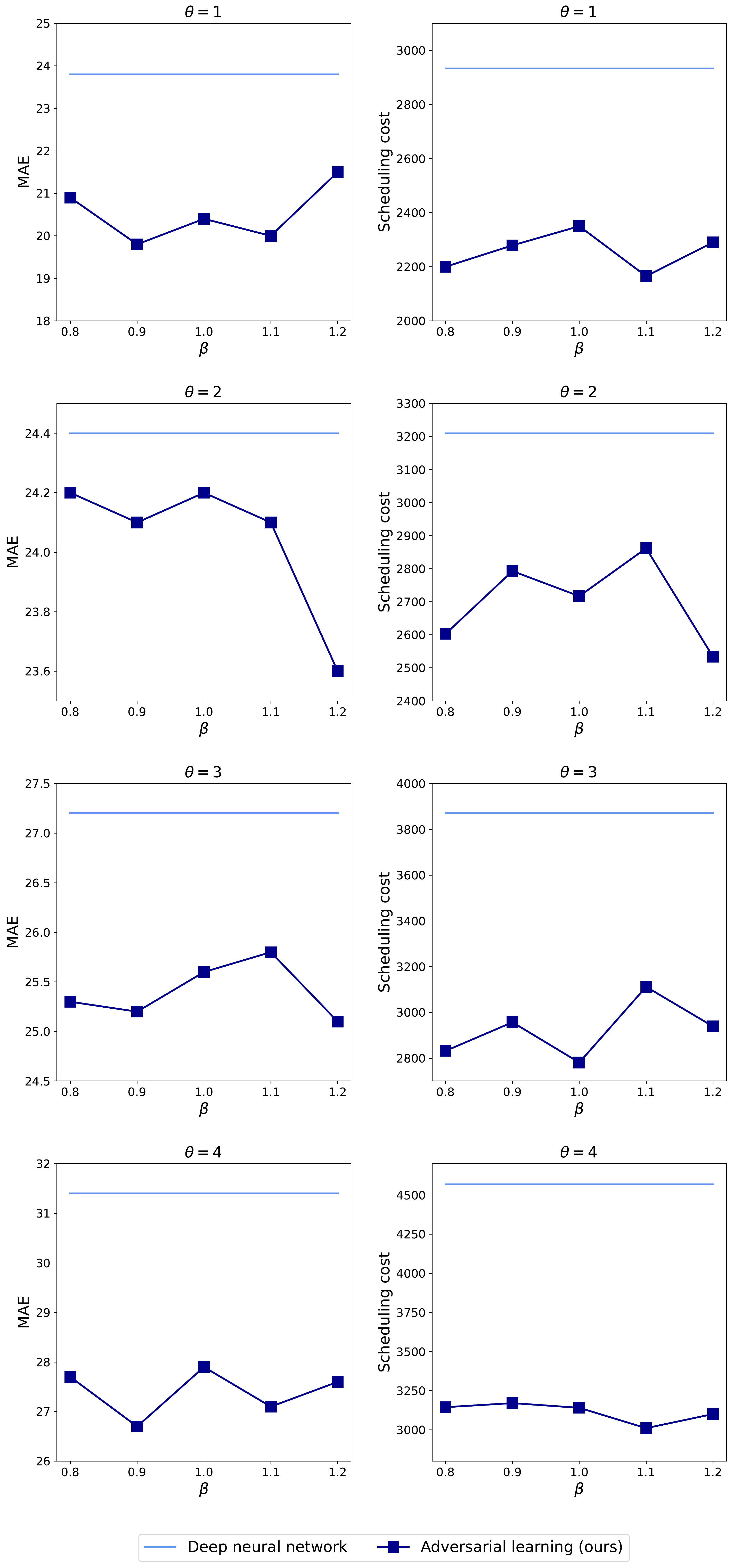}
\caption{Robustness checks for adversarial loss with respect to weight $\beta$.}
\label{fig:vary_beta}
\end{figure}

\clearpage 
\section{Comparison with machine learning baselines for handling distributional shifts}
\label{appendix:robustness_ML_baselines}

Here, we perform robustness checks, where we assume that the labels of the orders with throughput times that are within the first month of setting~\B are known (see \Cref{fig:baselines_timeline}). Hence, we assume a hypothetical scenario where, first, the data from setting~\B are acquired for one month by producing spools with short throughput times, and then, prediction and scheduling optimization is performed for the rest of the orders. Note that such a scenario is hardly applicable in practice because one does not know which spools have short throughput times, and we only use it here as a hypothetical construct to perform a robustness check. We use the main experimental setup from \Cref{sec:exp_setup}. Thereby, we rely on a rigorous hyperparameter search as described in Supplement~\ref{appendix:hyperparameter_tuning}, so that all comparisons are fair.

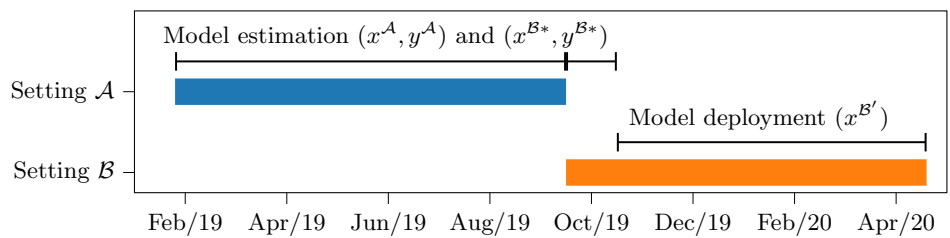
\begin{figure}
\centering
\begin{tikzpicture}
\footnotesize

\definecolor{color0}{rgb}{1,0.498039215686275,0.0549019607843137}
\definecolor{color1}{rgb}{0.12156862745098,0.466666666666667,0.705882352941177}

\begin{axis}[
width=0.75\textwidth,
height = 4 cm,
tick align=outside,
tick pos=left,
x grid style={white!69.0196078431373!black},
xmin=1, xmax=17,
xtick style={color=black},
y grid style={white!69.0196078431373!black},
ymin=-0.25, ymax=2,
xtick ={2,4,6,8,10,12,14,16},
xticklabels={Feb/19, Apr/19, Jun/19, Aug/19, Oct/19, Dec/19, Feb/20, Apr/20},
ytick ={0,1},
yticklabels={Setting $\mathcal{B}$, Setting $\mathcal{A}$},
ytick style={color=black}
]
\addplot [line width=10pt, color1]
table {%
1.8 1
9.5 1
};
\addplot [line width=10pt, color0]
table {%
9.5 0
16.6 0
};

\addplot [|-|,line width=0.75pt, black]
table {%
1.8 1.375
9.5 1.375
};

\addplot [|-|, line width=0.75pt, black]
table {%
9.5 1.375
10.5 1.375
};

\addplot [|-|, line width=0.75pt, black]
table {%
10.5 0.375
16.6 0.375
};

\draw (49.5,195) node [text width=10cm][align=center]{\baselineskip=10pt Model estimation $(x^\mathcal{A},y^\mathcal{A})$ and $(x^\mathcal{B*},y^\mathcal{B*})$ \par};

\draw (123,95) node [text width=4cm][align=center]{\baselineskip=10pt Model deployment $(x^\mathcal{B'})$ \par};

\end{axis}

\end{tikzpicture}
\caption{Timeline for numerical experiments comparing machine learning baselines (\ie, model retraining and transfer learning)}
\label{fig:baselines_timeline}
\end{figure}

The first baseline is \emph{model retraining} \citep{Cui.2018}. We implement this approach by training a deep neural network (\ie, without adversarial objective) based on the spool orders in setting~\A and spool orders in setting~\B for which we assume to know the throughput times (\ie, orders with throughput times that are less than one month). 

The second baseline is \emph{transfer learning} via fine-tuning \citep{McNamara.2017, Kouw.2018}. Here, we first train a deep neural network based on the spool orders in setting~\A. We then follow common practice in transfer learning by fine-tuning the weights of our neural network based on the orders with throughput times that are less than one month from setting~\B. 

The results of the comparison of the above baselines with our adversarial learning approach are reported in \Cref{tab:ml_baselines}. All predictions are eventually subject to the same optimization for job shop scheduling, analogous to the above numerical experiments. Our approach achieves lower prediction errors and lower scheduling costs as compared to these baselines, thereby demonstrating that our approach based on adversarial learning is superior. 

\begin{table}
\renewcommand*{\arraystretch}{1.05}
\SingleSpacedXI
\footnotesize
\centering
\caption{Main results for job shop scheduling.} 
\label{tab:ml_baselines}
\begin{tabular}{@{\extracolsep{4pt}} l cccc cccc}
\toprule
\textbf{Approach} & \multicolumn{4}{c}{Prediction error (MAE)} & \multicolumn{4}{c}{Scheduling cost} \\
\cline{2-5} \cline{6-9}
 & $\theta=1$ & $\theta=2$ & $\theta=3$ & $\theta=4$ & $\theta=1$ & $\theta=2$ & $\theta=3$ & $\theta=4$ \\
\midrule
Model retraining & 33.4 & 38.4 & 42.8 & 47.3 & 3379.5 &  3883.1 &  4321.2 & 4729.5 \\
 & ($\pm$1.8) & ($\pm$2.1) & ($\pm$3.9) & ($\pm$2.4) & ($\pm$199.6) & ($\pm$226.2) & ($\pm$400.7) & ($\pm$261.7) \\
Transfer learning &  49.0 & 48.5 & 48.6 & 50.4 & 4881.0 & 4780.1 & 4828.1  & 4966.6 \\
 & ($\pm$1.2) & ($\pm$1.6) & ($\pm$1.8) & ($\pm$2.1) & ($\pm$161.8) & ($\pm$189.1) & ($\pm$196.3) & ($\pm$229.1) \\
Adversarial learning (ours) & 23.5 & 23.7 & 24.0 & 23.8 & 2782.2 & 2802.7 & 2995.2 & 2954.2 \\
 & ($\pm$1.7) & ($\pm$1.4) & ($\pm$2.1) & ($\pm$1.5) & ($\pm$182.3) & ($\pm$183.9) & ($\pm$177.7) & ($\pm$176.2) \\
\bottomrule
\end{tabular}
\end{table}

\clearpage 
\section{Robustness in other operational contexts}
\label{appendix:robustness_other_settings}

Here, we report the results of our numerical experiment where we use a different operational context to show that our approach is transferable to other order settings. We again have two order settings: ``Johan Sverdrup Living Quarters Rig'' is used for training the throughout predictions in stage~1 of our approach (see \Cref{fig:order_images_robustness}), and ``Johan Castberg Floating Production Vessel'' is the order setting for which the job shop scheduling problem is solved. Both order settings have no chronological overlap (\ie, ``Johan Sverdrup Living Quarters Rig'' took place around 2016--2018, and ``Johan Castberg Floating Production Vessel'' in 2019). The rest of our experimental setup, including hyperparameter tuning, is identical to our main numerical experiment in \Cref{sec:exp_setup}. Specifically, both the deep neural network and our adversarial learning use exactly the {same} tuning grid, so that all performance gain must be attributed solely to the better learning objective (and not to a more flexible architecture). As a result, all comparisons are fair.

\begin{figure}
\centering
\captionsetup[sub]{font=scriptsize}
\begin{subfigure}{.35\textwidth}

  \centering
  \includegraphics[height=3cm]{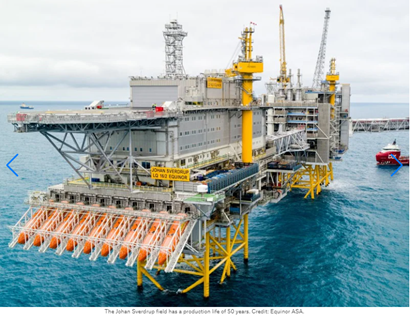}
\end{subfigure}

\vspace{0.01cm}
\caption{Operational context for robustness check: Johan  Sverdrup Living Quarters Rig (Credit: Equinor).} 
\label{fig:order_images_robustness}
\end{figure}

We report the results for prediction performance and scheduling costs (see \Cref{tab:other_setup}). The results show numerically that our approach is also effective in this operational context. In particular, we again observe the following: (1)~Our adversarial learning approach consistently outperforms off-the-shelf machine learning baselines for different magnitudes of the distributional shift in terms of both prediction performance (MAE) and realized scheduling cost. We also performed a Welch's $t$-test that compares our approach to the best off-the-shelf machine learning baseline. For $\theta = 2,3,4$, the improvements in scheduling costs are statistically significant at common significance thresholds. (2)~The performance gains of our approach increase as the magnitude of the distributional shift is increased. In sum, job shop scheduling using our adversarial learning approach is superior over job shop scheduling using off-the-shelf machine learning by a considerable margin. Thereby, we show the robustness of our results a different operational context at Aker. 

\begin{table}
\renewcommand*{\arraystretch}{1.05}
\SingleSpacedXI
\footnotesize
\centering
\caption{Results for a different operational context.} 
\label{tab:other_setup}
\begin{tabular}{@{\extracolsep{4pt}} l cccc cccc}
\toprule
& \multicolumn{4}{c}{Prediction error (MAE)} & \multicolumn{4}{c}{Scheduling cost} \\
\cline{2-5} \cline{6-9}
& $\theta=1$ & $\theta=2$ & $\theta=3$ & $\theta=4$ & $\theta=1$ & $\theta=2$ & $\theta=3$ & $\theta=4$ \\
\midrule
Linear regression (regularized) & 22.3 & 25.6 & 29.7  & 36.2 & 2833.6 & 3487.5 & 4089.4 & 5065.7 \\
 & ($\pm$0.8) & ($\pm$1.3) & ($\pm$3.3) & ($\pm$6.3) & ($\pm$103.3) & ($\pm$236.7) & ($\pm$422.8) & ($\pm$770.4) \\
Deep neural network & 27.4 & 34.4 & 47.6  & 66.2 & 3567.8 &  4622.5 & 6306.6 & 8182.3 \\
 & ($\pm$1.5) & ($\pm$4.5) & ($\pm$6.6) & ($\pm$10.7) & ($\pm$266.9) & ($\pm$565.8) & ($\pm$759.7) & ($\pm$763.2) \\
Adversarial learning (ours) & 22.2 & 24.5 & 24.6 & 23.3 & 2576.6 & 2782.7 & 2818.5 & 2701.6 \\
 & ($\pm$1.3) & ($\pm$1.3) & ($\pm$1.6) & ($\pm$2.4) & ($\pm$185.3) & ($\pm$161.0) & ($\pm$145.0) & ($\pm$252.6) \\
\midrule
Oracle (lower bound) & 0.0 & 0.0 & 0.0 & 0.0 & 379.2  & 363.3  & 373.5  & 403.1 \\
 & --- & --- & --- & --- & ($\pm$67.3) & ($\pm$69.6) & ($\pm$45.3)  & ($\pm$71.5) \\
\bottomrule
\end{tabular}
\end{table}

\end{document}